\documentclass[journal,12pt,onecolumn,draftclsnofoot]{IEEEtran}

\usepackage[letterpaper, left=1.01in, right=1.01in, bottom=1in, top=0.75in]{geometry}

\usepackage[utf8]{inputenc}

\normalsize
\DeclareMathAlphabet\mathbfcal{OMS}{cmsy}{b}{n}

\usepackage[english]{babel}
\usepackage{amsmath,amssymb}
\usepackage{cite}
\usepackage{graphicx}
\usepackage{dsfont}
\usepackage{amsthm}
\usepackage{amssymb}
\usepackage{mdframed}
\usepackage[disable,colorinlistoftodos,textwidth=\marginparwidth]{todonotes}
\newcommand{\idle}{\mathrm{i}}
\newcommand{\retx}{\mathrm{x}}
\newcommand{\new}{\mathrm{n}}
\usepackage{algorithm}
\usepackage{algorithmicx}
\usepackage{algcompatible}
\usepackage{eqparbox}
\usepackage{enumerate}

\makeatletter
\def\blfootnote{\xdef\@thefnmark{}\@footnotetext}
\makeatother

\usepackage{lscape}
\usepackage{bbm}
\DeclareMathOperator*{\argmin}{arg\,min} 
\DeclareMathOperator*{\argmax}{arg\,max} 
\newtheorem{theorem}{Theorem}
\newtheorem{problem}{Problem}

\newtheorem{definition}{Definition}
\newtheorem{lemma}{Lemma}

\newtheorem{proposition}{Proposition}

\newtheorem*{previous theorem}{Theorem~1 of \cite{journal_paper}}

\newcommand{\Exp}[1]{\mathbb{E}\left[ #1 \right]} 
\newcommand{\Var}[1]{\mathbb{V}\left[ #1 \right]} 
\newcommand{\Exppi}[1]{\mathbb{E}\left[ #1 \right]} 
\title{A Reinforcement Learning Approach \\ to Age of Information \\ in Multi-User Networks with HARQ}

\author{
\IEEEauthorblockN{Elif Tu\u{g}\c{c}e Ceran, Deniz G{\"u}nd{\"u}z, and Andr\'as Gy\"orgy}
}

\newcommand{\highlight}{}

\begin{document}

\maketitle

\vspace{-1cm}
\begin{abstract}
Scheduling the transmission of time-sensitive information from a source node to multiple users over error-prone communication channels is studied with the goal of minimizing the long-term average \textit{age of information (AoI)} at the users.  A long-term average resource constraint is imposed on the source, which limits the average number of transmissions. The source can transmit only to a single user at each time slot, and after each transmission, it receives an instantaneous ACK/NACK feedback from the intended receiver, and decides when and to which user to transmit the next update. Assuming the channel statistics are known, the optimal scheduling policy is studied for both the standard automatic repeat request (ARQ) and hybrid ARQ (HARQ) protocols. Then, a \textit{reinforcement learning} (RL) approach is introduced to find a near-optimal policy, which does not assume any \textit{a priori} information on the random processes governing the channel states.  Different RL methods including average-cost SARSA with linear function approximation (LFA), upper confidence reinforcement learning (UCRL2), and deep Q-network (DQN) are applied and compared through numerical simulations.
\smallskip

\textbf{Index Terms:} Age of information, hybrid automatic repeat request (HARQ), constrained Markov decision process, reinforcement learning, Whittle index.
\end{abstract}
\blfootnote{Part of this work is was presented at the  IEEE International Symposium on Personal, Indoor and Mobile Radio Communications, Bologna, Italy, September 2018 \cite{pimrc_paper}.
This work was supported in part by the European Research Council (ERC) Starting Grant BEACON (grant agreement no. 725731).
E.~T.~Ceran is with Imperial College London, UK, (She is currently with Middle East Technical University, Turkey, email: \texttt{elifce@metu.edu.tr}), D.~G\"und\"uz is with 
Imperial College London, UK (email: \texttt{d.gunduz@imperial.ac.uk}. A.~Gy\"orgy is with DeepMind, 
UK (email: \texttt{agyorgy@google.com}).}


\section{Introduction}

We consider a status update system, in which a source node wants to communicate the state of a time-varying process to multiple users.  The timeliness of the information at each user is measured by the \emph{age of information} (AoI), defined as the time elapsed since the most recent status update received by that user was generated at the source~\cite{Altman2010, Kaul2011, Kaul2012}. The goal of the source is to minimize the \emph{average} AoI across the users. Most of the earlier work on AoI consider queue-based models, in which the status updates arrive at the source node randomly according to a Poisson process, and are stored in a buffer before being transmitted to the destination \cite{Kaul2011, Kaul2012, Najm2017,Beytur2019}. Instead, we consider the so-called \emph{generate-at-will} model, in which the source can sample the process at any time and generate a fresh status update \cite{Altman2010,Tan2015, Sun2016, hsuage2017, Yates2017,wcnc_paper, journal_paper, Kadota2018, infocom_paper}.

\begin{figure}[!t]
\centering
\includegraphics[scale=0.5]{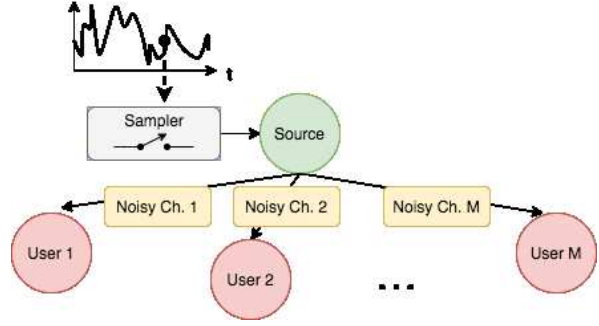}
\caption{The system model of a multi-user status update system over error prone links.}
\label{fig:system}
\end{figure}

We address the scheduling of status updates in a multi-user network under a transmission-rate constraint. This constraint is motivated by the fact that sensors sending status updates usually have limited energy supplies (e.g., are powered via energy harvesting \cite{Gunduz2014, Tan2015, infocom_paper}); hence, they cannot send an unlimited number of updates. We assume that the source can transmit to only a single user at each time slot, and the communication channels experience fading. While the source does not have channel state information, we assume the presence of a single bit perfect feedback link from each user to the source terminal, across which the corresponding receiver can send ACK/NACK feedback after each transmission. We consider both the standard ARQ and the hybrid ARQ (HARQ) protocols. Note that, in the former, the same transmission is repeated until it is successfully received; however, in a status update system no retransmission takes place, as it is always better to send a fresh status update. On the other hand, under HARQ, one may repeat previously sent packets as the probability of correct decoding increases with multiple transmissions. First, we assume that the success probability of each transmission attempt is known beforehand, in which case the source can judiciously decide when to transmit, or, in the case of HARQ, to retransmit or  discard failed information and send a fresh update. Then, we consider scheduling status updates over  unknown channels, in which case the success probabilities of transmission attempts are not known \textit{a priori}, and must be learned in an online fashion using the ACK/NACK feedback signals. 


AoI in multi-user networks has been studied in \cite{He2017,hsuage2017,Kadota2018, Yates2017,Kaul_multiaccess, Yates_multicast,bedewy2020}. It is shown in \cite{He2017} that the scheduling problem, where a set of links that share a common channel and the transmitter at each link contains a given number of packets
with time stamps from an information source, is NP-hard. Scheduling transmissions to multiple receivers is investigated in \cite{hsuage2017}, focusing on a perfect transmission medium, and the optimal scheduling algorithm is shown to be of threshold-type on the AoI. Average AoI has also been studied when status updates are transmitted over unreliable multiple-access channels \cite{Kaul_multiaccess} or multi-cast networks \cite{Yates_multicast}. A source node sending time-sensitive information to a number of users through unreliable channels is considered in \cite{Kadota2018}, where the problem is formulated as a \textit{restless multi-armed bandit} (RMAB), and a suboptimal Whittle Index (WI) policy is proposed. 


%

In \cite{Najm2017, Yates2017, Arafa2019}, AoI at a single user is studied when status updates are transmitted over an erasure channel with retransmissions. Two HARQ protocols are considered to combat erasures: infinite incremental redundancy (IIR) and fixed redundancy (FR) coding. The IIR protocol represents a system in which a status update packet is encoded with $k_s$ symbols ratelessly, such that the transmission of an update continues until $k_s$ symbols are received. An FR protocol represents an $(n_s, k_s)$-maximum distance separable (MDS) code, where each update is transmitted as an $n_s$ symbol packet, and the packet can be decoded if at least $k_s$ symbols are received. An information theoretic approach to the AoI problem is taken in~\cite{Najm2019}, where the optimal average AoI is characterized when no feedback is assumed.  In this paper, we instead consider scheduling of status updates to multiple users under a transmission rate constraint for general HARQ protocols, and we study standard ARQ and FR HARQ protocols as a special case of HARQ. In our earlier work, we studied a point-to-point status update system under a transmission-rate constraint \cite{wcnc_paper,journal_paper}, and showed that the optimal policy is a randomized stationary policy with randomization at most one state. As opposed to the single user setting, in the multi-user scenario considered in this paper, the source has to decide not only when to transmit, but also to which user to transmit, significantly increasing both the state and action spaces of the underlying problem.



Most prior literature on AoI assume perfect statistical knowledge of the random processes governing the status update system. However, in most practical systems (e.g., sensors embedded in unknown or time-varying environments), the characteristics of the system are not known \emph{a priori}, and must be learned. A limited number of recent works consider the unknown or time-varying characteristics of status update systems, and apply a learning-theoretic approach \cite{hsuage2017,wcnc_paper, pimrc_paper, journal_paper, infocom_paper,Beytur2019, Leng2019, Saad2019, Bennis2019}. The scheduling decisions with multiple receivers over a perfect channel is investigated in \cite{hsuage2017, Beytur2019}, where the goal is to learn data arrival statistics. Q-learning is used for a generate-at-will model in \cite{hsuage2017}, while policy gradients and DQN methods are used for a queue-based multi-flow AoI-optimal scheduling problem in~\cite{Beytur2019}. In~\cite{Leng2019}, policy gradients and
DQN methods are employed for AoI minimization in a wireless ad-hoc network, where nodes exchange status updates with one another over a shared spectrum. Average cost reinforcement learning (RL) algorithms are proposed in  \cite{journal_paper,pimrc_paper} to learn the decoding error probabilities in a status update system with HARQ. The work in \cite{infocom_paper} exploits RL methods in order to learn both the decoding error probabilities and the energy harvesting characteristics. 

To the best of our knowledge,  the average AoI with HARQ is studied for the first time for a multi-user system under a long-term average resource constraint. Similarly, there is no prior work in the literature which employs and compares the performances of the various RL methods exploited in this paper.  The main contributions of this paper can be summarized as follows:

\begin{itemize}
\item Both retransmission and pre-emption following a failed transmission are considered, corresponding, respectively, to the HARQ and ARQ protocols, and the structure of the optimal policy is determined.
\item The multi-user scheduling problem is shown to be indexable, and suboptimal WI policies are derived in closed-form for the standard ARQ and FR HARQ protocols.
\item Lower bounds on the average AoI are proposed for the standard ARQ and the FR HARQ protocols under a resource constraint.
\item We employ average-cost RL algorithms, in particular, \textit{average-cost SARSA}, \textit {upper confidence reinforcement learning} (UCRL2), \textit{average-cost SARSA with softmax and linear function approximation} (LFA) and \textit{deep reinforcement learning} (DRL) to learn the optimal scheduling decisions when the transmission success probabilities are unknown. 
\item Extensive numerical simulations are conducted in order to analyze the effect of the resource constraint, the network size, and the ARQ or HARQ mechanisms on the freshness of information, and the effectiveness of the proposed RL algorithms. 
\end{itemize}


\section{System Model and Problem Formulation}
\label{sec:system}



We consider a slotted status update system, where a source terminal monitors a time-varying process and sends updates about the process' state to multiple users.
In every time slot, the source terminal is able to generate an update at the beginning of the slot, and can transmit a status update to (at most) one of the $M$ users. 
This can be either because of dedicated orthogonal links to the users, for example, in a wired network, or because the users are interested in distinct information. A transmission attempt of a status update to a single user takes constant time, which is assumed to be equal to the duration of one time slot. 



We assume that the state of each of the channels changes randomly from one time slot to the next in an independent and identically distributed (i.i.d.) fashion, and the channel state information is available only at the corresponding receivers. We assume the availability of an instantaneous error-free single-bit ACK/NACK feedback from each user to the source. {\highlight  Successful reception of the status update at the end of time slot $t$ is acknowledged by an ACK signal (denoted  by $K_t= 1$),  while  a  NACK  signal is sent in case of a failure (denoted by $K_t= 0$).} In the standard ARQ protocol, a packet is retransmitted after each NACK feedback, until it is successfully decoded. However, in the AoI framework there is no point in retransmitting a failed out-of-date status packet if it has the same error probability as a fresh status update. Hence, the source always removes a failed status signal, and transmits a fresh update. On the other hand, in HARQ, signals from  previous transmission attempts are combined, and therefore the probability of error decreases with every retransmission~\cite{harq2003}.


{\highlight  In practice, the utility of status updates typically becomes zero beyond a certain age, hence we assume that the age is bounded; as such, we assume that the maximum age is $N<\infty$. Assuming that the most up-to-date packet received by the $j^{th}$ user ($j \in [M] \triangleq \{1,\ldots,M\}$) before time slot $t$ was generated in slot  $U_j(t)$, the AoI at the receiver of user $j$ at the beginning of time slot $t$ is defined as $\delta^{rx}_{j,t}\triangleq \min\{t-U_j(t),N\}\in [N] \triangleq \{1,\ldots,N\}$.}

{\highlight At each time slot $t$, the source node takes an action $a_t$ from the set of actions $\mathcal{A}=\{\idle,\new_1,\retx_1,\ldots,$ $ \new_M,\retx_M\}$: in particular, the source can i) remain idle ($a_t=\idle$); ii) generate and transmit a new status update to the $j^{th}$ user ($a_t=\new_j$, $j\in [M]$); or, iii) retransmit the most recent failed status update to the $j^{th}$ user ($a_t=\retx_j$, $j\in [M]$). Note that $|\mathcal{A}|=2M+1$. For the $j^{th}$ user, the probability of error after $r$ retransmissions, denoted by $g_j(r)$, depends on $r$ and the particular HARQ scheme used \cite{harq2003}. In any reasonable HARQ strategy, $g_j(r)$ is non-increasing in $r$, i.e., $1>g_j(r) \geq g_j(r')> 0$ for all $r \leq r'$.  We will denote the maximum number of retransmissions by $r_{max}$. We note that standard HARQ methods only allow a finite maximum number of retransmissions (e.g., $r_{max} = 3$ \cite{IEEEstandard,LTEmax}).

Let $\delta^{tx}_{j,t}$ denote the number of time slots elapsed since the generation of the most recently transmitted (successfully or not) packet to user $j$ at the transmitter, while recall that $\delta^{rx}_{j,t}$ denote the AoI of the most recently received status update at the receiver of user $j$. $\delta^{tx}_{j,t}$ resets to 1 if a new status update is generated for user $j$ at time slot $t-1$, and increases by one (up to $N$) otherwise, i.e., 
\begin{align*}
\delta^{tx}_{j,t+1}=
\begin{cases}
1 &\textrm{ if }  a_t=\new_j; \\
\min(\delta^{tx}_{j,t}+1,N) &\textrm{ otherwise. }
\end{cases}
\end{align*}
On the other hand, the AoI at the receiver side evolves as follows:
\begin{align*}\delta^{rx}_{j,t+1}=
\begin{cases}
1 &\textrm{if } a_t=\new_j \textrm{ and } K_t=1; \\
\min(\delta^{tx}_{j,t}+1,N) &\textrm{if } a_t=\retx_j \textrm{ and } K_t = 1; \\
\min(\delta^{rx}_{j,t}+1,N) & \textrm{otherwise. }
\end{cases}
\end{align*}
Note that once the AoI at the receiver is at least as large as at the transmitter, this relationship holds forever; thus it is enough to consider cases when $\delta^{rx}_t \ge \delta^{tx}_t$.

Therefore, $\delta^{rx}_{j,t}$ increases by 1 when the source chooses to transmit to another user, or if the transmission fails, while it decreases to 1, or, in the case of HARQ, to $\min(\delta^{tx}_{j,t}+1,N)$, when a status update is successfully decoded. Also, $\delta^{tx}_{j,t}$ increases by $1$ if the source chooses not to generate a new packet and transmit it to user $j$ ($a_t\neq \new_j$). 

For the $j^{th}$ user, let  $r_{j,t}\in \{0,\ldots,r_{max}\}$ denote the number of previous transmission attempts of the most recent packet. Thus, the number of retransmissions is zero for a newly sensed and generated status update and increases up to $r_{max}$ as we keep retransmitting the same packet. Then, the state of the system can be described by the vector $s_t \triangleq (\delta^{rx}_{1,t}, \delta^{tx}_{1,t}, r_{1,t}, \ldots, \delta^{rx}_{M,t}, \delta^{tx}_{M,t}, r_{M,t})$, where $s_t$ belongs to the set of possible states $\mathcal{S} \subset ([N] \times [N] \times [r_{max}])^M$.  }

If no resource constraint is imposed, remaining idle is clearly a suboptimal action. However, in practice, continuous transmission is typically not possible due to energy or interference constraints. To model these situations, we impose a constraint on the average number of transmissions, denoted by $\lambda \in (0,1]$.
{\highlight 
This leads to a constrained Markov decision proccess (CMDP) formulation, defined by the 5-tuple $\big(\mathcal{S}, \mathcal{A},  \mathcal{P}, c, d\big)$: The countable set of states  $\mathcal{S} $ and the finite set of actions $\mathcal{A}$ have already been defined. $\mathcal{P}$ refers to the transition kernel and can be summarized as follows:
\begin{align}
\mathcal{P}_{s,s'}(a)=
\begin{cases}
 1  &\mathrm{ if } ~a=\idle, ~\delta^{rx'}_i=\min\{\delta^{rx}_i+1,N\}, ~\delta^{tx'}_i=\min\{\delta^{tx}_i+1,N\},
\\ &~r'_i=r_i, ~\forall i;   \\ 
1-g_j(0)   &\mathrm{ if } ~a=\new_j,  ~\delta^{rx'}_j=1, ~\delta^{tx'}_j=1, ~r'_j=0, 
~\delta^{rx'}_i=\min\{\delta^{rx}_i+1,N\}, \\ &~\delta^{tx'}_i=\min\{\delta^{tx}_i+1,N\}, ~r'_i=r_i, ~\forall i \neq j; \\ 
g_j(0)   &\mathrm{ if } ~a=\new_j, ~\delta^{rx'}_j=\min\{\delta^{rx}_j+1,N\}, ~\delta^{tx'}_j=1, ~r'_j=1, ~r'_i=r_i; 
\\ &~\delta^{rx'}_i=\min\{\delta^{rx}_i+1,N\}, ~\delta^{tx'}_i=\min\{\delta^{tx}_i+1,N\}, \forall i \neq j; \\ 
1-g_j(r_j)   &\mathrm{ if } ~a=\retx_j, ~\delta^{rx'}_j=\delta^{tx}_j+1,
~\delta^{tx'}_j=\min\{\delta^{tx}_j+1,N\},  ~r'_j=0, ~r'_i=r_i,
\\ &~\delta^{rx'}_i=\min\{\delta^{rx}_i+1,N\}, ~\delta^{tx'}_i=\min\{\delta^{tx}_i+1,N\},  \forall i \neq j; \\ 
g_j(r_j)   &\mathrm{ if } ~a=\retx_j, ~\delta^{rx'}_j=\min\{\delta^{rx}_j+ 1, N\}, ~\delta^{tx'}_j=\min\{\delta^{tx}_j+1,N\},
 \\ &~\delta^{rx'}_i=\min\{\delta^{rx}_i+1, N\}, ~\delta^{tx'}_i=\min\{\delta^{tx}_i+1,N\},  \\ &~r'_j = \min\{r'_j + 1,r_{max}\}, ~r'_i = r_i,
 \forall i \neq j;\\
0 & \text{otherwise}, \\[0.75em]
\end{cases}\label{eq:transitions}
\end{align}
where  $\mathcal{P}_{s,s'}(a) = \Pr(s_{t+1}\!=\!s' \mid s_t\!=\!s, a_t\!=\!a)$ is the probability that action $a \in \mathcal{A}$  in state $s\in\mathcal{S}$  at time $t$  leads to state $s'\in\mathcal{S}$ at time $t+1$ (the components of state $s'$ are denoted by a prime in the above equation). The instantaneous cost function $c: \mathcal{S} \times \mathcal{A} \rightarrow \mathbbm{R}$ is defined as the weighted sum of the AoIs at the multiple users, independently of $a$. Formally, $c(s,a)=\Delta \triangleq w_1\delta^{rx}_{1}+\cdots+ w_M\delta^{rx}_{M}$, where the weight $w_j>0$ represents priority of user $j$. The instantaneous transmission cost $d:\mathcal{A} \to \mathbb{R}$ is defined as $d(\idle)=0$ and $d(a)=1$ if $a\neq\idle$.
}

Naturally, as reflected by the system model, for every user we keep only the most recent status update packet: thus, the number of retransmissions is zero for a newly sensed and generated status update and increases up to $r_{max}$ as we keep retransmitting the same packet. {\highlight  
If a maximum of $r_{max}$ retransmissions is reached, the packet can still be retransmitted; however, due to the protocol, only the last $r_{max}$ retransmissions are used in the decoding, hence the retransmission count saturates at $r_{max}$. Figure~\ref{fig:example} illustrates an example showing the actions and state transitions for a 2-user system.


\begin{figure}
    \centering
    \includegraphics[scale=0.35]{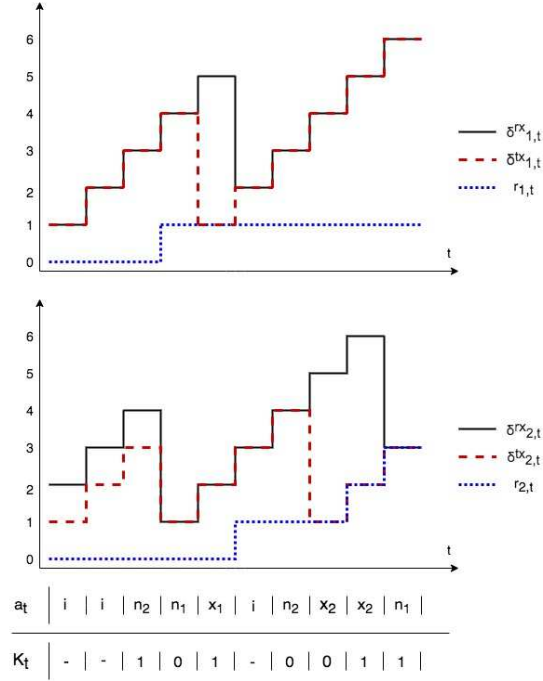}
    \vspace{-0.2in}
    \caption{An example illustrating the AoIs and retransmission numbers for a 2-user network in the presence of ACK/NACK feedback}
    \label{fig:example}
\end{figure}



A stationary \emph{policy} $\pi: \mathcal{S} \times \mathcal{A} \rightarrow [0,1]$ maps each state $s\in \mathcal{S}$ to an action $a\in \mathcal{A}$ with probability $\pi(a|s)$ ($\pi(\cdot|s)$ is a distribution over $\mathcal{A}$). 
We use $s_t^{\pi}=({\delta^{rx}_{1,t}}^{\pi},{\delta^{tx}_{1,t}}^{\pi},r_{1,t}^{\pi},\ldots,{\delta^{rx}_{M,t}}^{\pi},{\delta^{tx}_{M,t}}^{\pi},r_{M,t}^{\pi})$ and $a_t^{\pi}$ to denote the sequences of states and actions, respectively, induced by policy $\pi$, while $\Delta^{\pi}_t\triangleq\sum_{j=1}^M {w_j{\delta^{rx}_{j,t}}^{\pi}}$ denotes the instantaneous weighted cost. }

The infinite horizon expected weighted average AoI for policy $\pi$ starting from the initial state $s_0 \in \mathcal{S}$ is defined as
\vspace{-0.2in}
\begin{align}
J^{\pi}(s_0) \triangleq \limsup_{T\rightarrow \infty }\frac{1}{T}\Exppi{\sum_{t=1}^T{\Delta^{\pi}_t}\Big|s_0},
\end{align}
while the corresponding average number of transmissions is given by
\begin{align}
C^{\pi}(s_0) \triangleq \limsup_{T\rightarrow \infty }\frac{1}{T}\Exppi{\sum_{t=1}^T{\mathbbm{1}[a^{\pi}_t \neq \idle]  }\Big|s_0}~.\label{eq:constraint}
\end{align}

We are interested in minimizing $J^{\pi}(s_0)$ given a constraint $\lambda$ on the average number of transmissions $C^{\pi}(s_0)$, leading to the following CMDP optimization problem:
\begin{problem} \rm
$\underset{\pi\in \Pi}{\mathrm{ Minimize }}  ~J^{\pi}(s_0)$ over $\pi\in \Pi$ such that
$C^{\pi}(s_0) \le \lambda.$ 
\label{problem}
\end{problem} 
Without loss of generality, we assume that the state at the beginning of the problem is $s_0=(1,1,0,2,1,0,\ldots,M,1,0)$; and we omit $s_0$ from the notation for simplicity. A policy $\pi^*\in \Pi$ is called optimal if $J^*\triangleq J^{\pi^*} \leq J^{\pi}$ for all $\pi\in \Pi$ and we are interested in finding optimal policies.

\section{Lagrangian Relaxation and the Structure of the Optimal Policy}
\label{sec:structure}

A detailed treatment of finite-state finite-action discounted MDPs is considered in \cite{Puterman_book}, but here we need more general results that apply to MDPs and CMDPs with average expected cost~\cite{Puterman_book,Altman}. 
Below we follow \cite{Altman} and \cite{Ross_1985} to characterize the optimal policy.

We will need two well-known concepts for MDPs \cite{Puterman_book,Altman}: An MDP is \emph{communicating} if for any two states $s,s'$ there exists a deterministic policy $\pi$ such that $s'$ is reachable from $s'$ with positive probability following $\pi$. A stronger concept is the \emph{unichain} property, which we define for the more general class of CMDPs: a finite CMDP is unichain if any feasible policy (i.e., a policy that satisfies the resource constraint) induces a finite-state Markov chain that contains a single recurrent class and possibly, some transient states. We will show below that our MDP is communicating (cf. Theorem~\ref{thm:comm}) and that it is unichain under the ARQ protocol (cf. Theorem~\ref{thm_mixture}).



 To solve the constrained MDP, we start by rewriting Problem~\ref{problem} in its Lagrangian form. The average Lagrangian cost of a policy $\pi$ with Lagrange multiplier $\eta \ge 0$, denoted by $L^{\pi}_{\eta}$, is defined as
 \vspace{-0.1in}
\begin{align}
L^{\pi}_{\eta}=\lim_{T\rightarrow \infty }\frac{1}{T}\left(\Exp{\sum_{t=1}^T{\Delta^{\pi}_t}}\!+\!\eta \Exp{\sum_{t=1}^T{\mathbbm{1}[a^{\pi}_t\neq \idle]}}\right) \label{eq:lag}
\end{align}
and, for any $\eta$, the optimal achievable cost is defined as $L_{\eta}^*\triangleq \inf_{\pi}{L^{\pi}_{\eta}}$.
This formulation is equivalent to an unconstrained countable-state average-cost MDP with instantaneous (overall) cost  $\Delta_t^\pi+\eta\mathbbm{1}[a^{\pi}_t \neq \idle]$.


If $\lambda=1$, a transmission (new update or retransmission) is allowed in every time slot, and instead of a CMDP we have a finite-state MDP with bounded cost. Then it follows directly from Theorem~8.4.3 and Theorem~8.4.5 of~\cite{Puterman_book} that if the MDP is unichain (which is the case for the ARQ protocol as shown in Theorem~\ref{thm_mixture}), there exists an optimal deterministic policy that satisfies the well-known Bellman equations. In this section, we focus on the more interesting constrained problem. The constraint on the transmission cost is less than or equal to one (i.e., $\lambda\leq 1$), then we have $\eta\geq 0$, which will be assumed throughout the paper. A policy $\pi$ is called $\eta$-optimal if it achieves $L^*_\eta$.

\begin{theorem}
\label{thm:comm}
An optimal stationary policy $\pi^*_n$ minimizing \eqref{eq:lag} (and hence achieving $L_{\eta}^*$) exists for the unconstrained MDP with Lagrangian parameter $\eta$. 
\end{theorem}
\begin{proof}
First, we show that the unconstrained MDP is communicating, that is, for every pair of $(s,s')\in\mathcal{S}$, there exists a deterministic policy under which $s'$ is accessible from $s$. 
It is easy to see that there exists a policy which induces a recurrent Markov chain: Consider the policy which always transmits to the user with the smallest index such that the corresponding AoI at the user is less than $N$ or the retransmission count is less than $r_{max}$, sending a new packet if the retransmission count is 0 and retransmitting if it is not. This policy gets to the state $(N,N,r_{max},\ldots,N,N,r_{max})$ from any other state with at least a fixed positive probability in at most $M \max\{N,r_{max}\}$ steps, hence it induces a recurrent Markov chain. It follows than from Proposition 8.3.1 of~\cite{Puterman_book} that the MDP is communicating.
Then, by Theorem 8.3.2 of~\cite{Puterman_book}, an optimal stationary policy satisfying \eqref{eq:Bellman} exists.
\end{proof}

On the other hand, if the MDP is unichain, we can obtain stronger results describing the structure of the optimal policy. In this case, there exists a function $h_{\eta}(s)$, called the \textit{differential cost function}, satisfying the so-called \emph{Bellman optimality} equations
\vspace{-0.1in}
\begin{equation}
\label{eq:Bellman}
h_{\eta}(s)+L^*_\eta=\min_{a\in\mathcal{A}}\big(\Delta+\eta \cdot \mathbbm{1}[a \neq \idle]+\Exp{h_{\eta}(s')|s,a}\big),~ \forall s \in \mathcal{S}, 
\end{equation}
where $s'\in\mathcal{S}$ is the next state obtained from $s$ after taking action $a$ \cite{Puterman_book}. 
We also introduce the \textit{state-action cost function} defined as
\vspace{-0.1in}
\begin{align}
Q_{\eta}(s,a)\triangleq \Delta+\eta \cdot \mathbbm{1}[a \neq \idle]+\Exp{h_{\eta}(s')|s,a}, ~\forall s\in \mathcal{S}, a \in \mathcal {A}.
\label{eq:Bellman2} \vspace{-0.1in}
\end{align} 
Then, at each state $s\in\mathcal{S}$, the optimal deterministic policy takes the action achieving the minimum in \eqref{eq:Bellman2}:
\vspace{-0.1in}
\begin{align}
\label{eq:opt_eta}
\pi_{\eta}^*(s) &\in \argmin_{a\in\{\idle,\new,\retx\}} Q_{\eta}(s,a)~. 
\end{align}

{\highlight 
For a single-user point-to-point status update system, \cite{journal_paper} characterizes the structure of the optimal policy, and shows that there exists a stationary policy which randomizes in at most one state. 
Next we extend this result to multi-user status update systems for the ARQ protocol.

If we assume that the system adopts the standard ARQ protocol, that is, failed transmissions are discarded at the destination, then the state space reduces to  $(\delta^{rx}_1, \delta^{rx}_2, \ldots, \delta^{rx}_M)$ as $r_{j,t}=0,~\forall j,t$, and the action space to $\mathcal{A}=\{\idle , \new_1, \ldots, \new_M\}$. The probability of error of each status update is $p_j\triangleq g_j(0)$ for user $j$. 
State transitions in \eqref{eq:transitions} and the Bellman optimality equations can all be modified accordingly. Then we can extend Theorem~1 of \cite{journal_paper} to multi-user systems.

\begin{theorem}
There exists an optimal stationary policy for Problem~\ref{problem} under standard ARQ, which is optimal for the unconstrained problem considered in \eqref{eq:lag} for some $\eta=\eta^*$, and randomizes in at most one state. This policy can be expressed as a mixture of two deterministic policies $\pi^*_{\eta^*,1}$ and $\pi^*_{\eta^*,2}$ that differ in at most a single state $\hat{s}$, and are both optimal for the Lagrangian problem \eqref{eq:lag} with $\eta=\eta^*$. More precisely, there exist two deterministic policies $\pi^*_{\eta^*,1}$, $\pi^*_{\eta^*,2}$ as described above and $\mu \in [0,1]$, such that the mixture policy $\pi^*_{\eta^*}$, which selects, in state $\hat{s}$, $\pi^*_{\eta^*,1}(\hat{s})$ with probability $\mu$ and $\pi^*_{\eta^*,2}(\hat{s})$ with probability $1-\mu$, and otherwise follows these two policies (which agree in all other states)
is optimal for Problem~\ref{problem}, and the constraint in \eqref{eq:constraint} is satisfied with equality. \label{thm_mixture}
\end{theorem}
\begin{proof}
Since the state $(N,N,\ldots, N)$ is visited under every stationary policy with at least a fixed positive probability in at most $N M$ steps from every other state under the ARQ protocol (if $NM$ transmissions fail), the CMDP is unichain.
Then, by Theorem~4.4 of~\cite{Altman}, since Problem~\ref{problem} is feasible (i.e., there exists at least one policy which satisfies the constraint~\eqref{eq:constraint}), there exists an optimal stationary policy that is a mixture of two deterministic policies that differ in at most a single state with $\mu \in [0,1]$. From Section 4.4,  Theorem 3.6 and Theorem 4.4 of~\cite{Altman}, the mixture policy $\pi^*_{\eta^*}$, for any $\mu \in [0,1]$, also satisfies \eqref{eq:Bellman}, and is optimal for the unconstrained problem in \eqref{eq:lag} with $\eta=\eta^*$.
This completes the proof of the theorem.
\end{proof}

}

Some other results in \cite{Altman,Ross_1985}  will be useful in determining $\pi^*_{\eta^*}$. For any $\eta>0$, let $C_{\eta}$ and $J_{\eta}$ denote the average number of transmissions and average AoI, respectively, for the optimal policy $\pi_{\eta}^*$. Note that, $C_\eta$ and $J_\eta$ can be computed directly by finding the stationary distribution of the chain, or estimated empirically by running the MDP with policy $\pi^*_\eta$. 



To determine the optimal policy, one needs to find $\eta^*$, and the policies $\pi^*_{\eta^*,1}$ and $\pi^*_{\eta^*,2}$, 
In fact, \cite{Ross_1985} shows that $\eta^*$ is defined as 
\vspace{-0.2in}
\begin{align}
\eta^* \triangleq \inf\{\eta>0:C_{\eta}\le \lambda\},
\label{eq:eta_star}
\end{align}
where the inequality $C_{\eta}\le \lambda$ is satisfied if it is satisfied for at least one
of $C^{\pi^*_{\eta^*,i}}$ for $i=1$ or $i=2$. By Lemma 3.3~of \cite{Ross_1985}, $\eta^*$ is finite, and $\eta^*>0$ if $\lambda<1$.


Theorem~\ref{thm_mixture} and the discussion above describe the general structure of the optimal policy.  A detailed discussion on finding both $\eta^*$ and the policies $\pi^*_{\eta^*,1}$ and $\pi^*_{\eta^*,2}$ are given in \cite{journal_paper}, which is not the focus of this paper.  In Section~\ref{sec:solution}, for practical implementation,  an iterative heuristic algorithm, also is employed based on the discussion in this section. 



\section{An Iterative Algorithm to Minimize AoI}
\label{sec:solution}


For a practical solution to  our problem, we can employ the \emph{relative value iteration} (RVI) \cite{Puterman_book} algorithm to solve \eqref{eq:Bellman} for any given $\eta$; and hence, find (an approximation of) the optimal policy $\pi^*_\eta$. To simplify the notation, the dependence on $\eta$ is suppressed in the algorithm for $h$ and $Q$. Note that a finite approximation is needed for the practical implementation of the RVI algorithm since each iteration of RVI requires the computation of the value function for each state-action pair. 

{\highlight  The RVI algorithm essentially computes the optimal state value function through improving the estimates of state and state action values, $h(s)$ and $Q(s,a)$, respectively. The algorithm starts with a random initialization of $h_0(s)$, $\forall s$, and sets an arbitrary but fixed reference state $s^{ref}$.  It then continuously updates the $Q(s,a)$ and $h(s)$ values until convergence. A single iteration of the RVI algorithm is given as follows: }     
\begin{align}
Q_{n+1}(s,a) &\leftarrow\Delta+\eta \cdot \mathbbm{1}[a^{\pi} \neq \idle]+	\Exppi{h_n(s')|a},\\
 h_{n+1}(s) &\leftarrow \min_{a}(Q_{n+1}(s,a))-\min_{a}(Q_{n+1}(s^{ref},a)),
\end{align} 
where $Q_n(s,a)$ and $h_n(s)$ denote the state action value function and differential value function for iteration $n$, respectively.

\begin{algorithm}[h]
\begin{footnotesize}
\begin{algorithmic}[1]
 \renewcommand{\algorithmicrequire}{\textbf{Input:}}
 \renewcommand{\algorithmicensure}{\textbf{Output:}}
 \REQUIRE Lagrange parameter $\eta$, error probability  $g(r)$, $(\delta^{ref},r^{ref})$ \Comment{choose an arbitrary but fixed reference state}
  \STATE  $h_0^{N\times r_{max}}\leftarrow \mathbf{0}~$ \Comment{initialization}   
  \FOR {episodes $n= 0,1,2,\ldots$ } 
  \FOR {state $s\in \mathcal{S} $ }
  \FOR{action $a\in\mathcal{A}$}   
 \STATE	$Q_{n+1}(s,a)\leftarrow\Delta+\eta\cdot\mathbbm{1}[a^{\pi}\neq \idle]+	\Exp{h_n(s')}$\;  
\ENDFOR    
\STATE ${V}_{n+1}(\delta,r)\leftarrow\min_{a}(Q_{n+1}(\delta,r,a))$\;
\STATE $h_{n+1}(\delta,r)\leftarrow{V}_{n+1}(\delta,r)-{V}_{n+1}(\delta^{ref},r^{ref})$
  \ENDFOR
  \IF{$|h_{n+1}-h_n|\leq \epsilon$} \Comment{check the convergence}
  \FOR{$s \in \mathcal{S}$}    \Comment{compute the optimal policy}
\STATE	 $\pi^*_{\eta}(s)\leftarrow\argmin_{a}(Q(s,a))$ 

    \ENDFOR
    \STATE \textbf{Return} $\pi^*$ 
  \ENDIF

  \ENDFOR
 \end{algorithmic} 
  \caption{Relative Value Iteration (RVI)}
  \label{alg:RVI}
  \end{footnotesize}
\end{algorithm}

After presenting an algorithm that can compute the optimal deterministic policy $\pi^*_\eta$ for any given $\eta$ (more precisely, an arbitrarily close approximation thereof for the finite-state MDP), we need to find the particular Lagrange multiplier $\eta^*$ as defined by \eqref{eq:eta_star}. A heuristic method to find a single $\eta$ value with $C_{\eta}\approx\lambda$ is as follows: We start with an initial parameter $\eta^0$, and run an iterative algorithm updating $\eta$ as $\eta^{m+1} = \eta^m+\alpha_m (C_{\eta^m}-\lambda)$ for a step size parameter $\alpha_m$\footnote{$\alpha_m$ is a positive decreasing sequence and satisfies the following conditions: $\sum_m \alpha_m = \infty$ and $\sum_m \alpha^2_m < \infty$ from the theory of stochastic approximation \cite{stochastic_approx}.}. We continue this iteration until $|C_{\eta^{m}}-\lambda|$ becomes smaller than a given threshold, and denote the resulting value by ${\eta}^*$. A more detailed discussion on an iterative algorithm minimizing AoI is also given in~\cite{wcnc_paper,journal_paper}. 



\section{AoI with Standard ARQ Protocol}
\label{sec:arq}


In this section, we assume that the system adopts the standard ARQ protocol and the state space reduces to  $(\delta^{rx}_1, \delta^{rx}_2, \ldots, \delta^{rx}_M)$ as $r_{j,t}=0,~\forall j,t$, and the action space to $\mathcal{A}=\{\idle , \new_1, \ldots, \new_M\}$. The probability of error of each status update is $p_j\triangleq g_j(0)$ for user $j$. State transitions in \eqref{eq:transitions}, the Bellman optimality equations, and the RVI algorithm can all be simplified accordingly.  Thanks to these simplifications, we are able to derive a low-complexity policy based on Whittle's approach \cite{whittle1988}  by modelling the problem as a RMAB \cite{Gittins_book}. 
 Although the RVI algorithm presented in Section~\ref{sec:solution} provides an optimal solution to Problem~\ref{problem}, its computational complexity is significant for large networks consisting of many users. The WI policy in Section~\ref{sec:WI_arq}, on the other hand, results in a possibly suboptimal yet computationally efficient policy, which often performs very well in practice. We also derive a closed-form lower bound for the constrained MDP in Section~\ref{sec:bound}.




\subsection{WI Policy}
\label{sec:WI_arq}


Multi-armed bandits (MAB) \cite{Gittins_book} constitute  a class of RL problems with a single state. In the restless MAB (RMAB) problem \cite{whittle1988}, each arm is associated with a state that evolves over time, and the reward distribution of the arm depends on its state (in contrast, in the basic stochastic MAB problems, rewards are i.i.d.).
The multi-user AoI minimization problem with ARQ can be formulated as a RMAB with $M+1$ arms: choosing arm $j$ is associated with transmitting to user $j$, while arm $M+1$ represents the action of staying idle ($a=\idle$). RMAB problems are known to be PSPACE-hard in general \cite{Gittins_book}; however, a low-complexity heuristic policy  can be found for certain problems by relaxing the constraint that in every round only a single arm can be selected, and instead introducing a bound on the expected number of arms chosen~\cite{whittle1988}. The resulting policy, known as the WI policy, is a sub-optimal policy, but it is known to perform close to optimal in many settings \cite{whittle1988}. 

Following Whittle's approach, we decouple our problem  into $M$ sub-problems each corresponding to a single user, and treat these problems independently. The cost of transmitting to a user (called \textit{subsidy for passivity} \cite{whittle1988}) is denoted by $C$, which will be later used to derive the index policy. Writing the Bellman equation  \eqref{eq:Bellman2} for each subproblem, we obtain the optimality equations for the single user AoI minimization problem with the standard ARQ protocol where the action space is $\{\idle,\new_j\}$ 
\begin{equation}
h_C(\delta^{rx}_j)+L_j^*=\min \big\{Q(\delta^{rx}_j,\new_j),Q(\delta^{rx}_j,\idle)\}, \label{eq:BellmanC}
\end{equation} and the optimal policy to each subproblem is given
\begin{equation}
\pi_{C}^*(\delta^{rx}_j) \in \argmin_{a\in\{\idle,\new_j\}} \big\{Q(\delta^{rx}_j,a)\big\}, \textrm{ where}
\label{eq:BellmanC2}
\end{equation}
\vspace{-0.3in}
\begin{align*}
&Q(\delta^{rx}_j,\new_j)\triangleq w_j\delta^{rx}_j+ C + p_j h_{C}(\delta^{rx}_j+1) +(1-p_j) h_{C}(1), 
&Q(\delta^{rx}_j,\idle)\triangleq w_j\delta^{rx}_j+ h_{C}(\delta^{rx}_j+1).
\end{align*}


Given \eqref{eq:BellmanC} and \eqref{eq:BellmanC2}, let $S_j^{\new_j}(C)$ represent the set of states the optimal action is equal to $\new_j$ for a given $C$, that is, $S_j^{\new_j}(C)=\{s: \pi_{C}^*(\delta^{rx}_j)=\new_j\}$. Then, we define indexability as follows.

\begin{definition} An arm is indexable if the set $S_j^{\new_j}(C)$ as a function of $C$ is monotonically decreasing for $C\in \mathbb{R}$, and $ \underset{C\rightarrow \infty}\lim S_j^{\new_j}(C) =  \varnothing $  and $\underset{C\rightarrow -\infty}\lim S_j^{\new_j}(C) =  \mathcal{S}$ \cite{whittle1988,Gittins_book}. The problem is indexable if every arm is indexable.
\label{def:index}
\end{definition}


Note that if a problem is indexable as defined in Definition~\ref{def:index}, $S_j^a(C_1)\subset S_j^a(C_2)$ for $C_1\geq C_2$, and there exists a $C$ such that both actions are \emph{equally desirable}, that is, $Q(\delta^{rx}_j,\idle)=Q(\delta^{rx}_j,\new_j)$ for all $\delta^{rx}_j$. The WI for our problem is defined as follows.

\begin{definition}
\textit{The WI  for user $j$ at state $\delta^{rx}_j$,  denoted by $I_j(\delta^{rx}_j)$, is defined as the cost $C$ that makes both actions $\new_j$ and $\idle$ equally desirable.} 
\label{def:whittle}
\end{definition}
\vspace{-0.1in}
Next, we derive the WI for our problem: 
\vspace{-0.1in}
\begin{proposition}
Problem~\ref{problem} with standard ARQ is indexable and the WI for each user $j$ and state $\delta^{rx}_j$ can be computed as
\vspace{-0.1in}
\begin{equation}
I_j(\delta^{rx}_j)= \frac{1}{2} w_j \delta^{rx}_j (1-p_j)\left(\delta^{rx}_j+\frac{1+p_j}{1-p_j}\right), ~\forall j \in [M], \label{eq:WI}
\end{equation}
where the WI for the idle action is $I_{M+1}=\eta$. 
\label{prop_WI}
\end{proposition}
\vspace{-0.2in}
\begin{proof}
The proof is given in Appendix~\ref{append_WI}.
\end{proof}


 The WI policy  is defined as follows: in state $(\delta^{rx}_1,\delta^{rx}_2,\ldots,\delta^{rx}_M)$, compare the highest index with the Lagrange parameter $\eta$, and if $\eta$ is smaller,  then the source transmits to the user with the highest index, otherwise the source remains idle. The WI policy, defined below, tends to transmit to the user with a high weight ($w_j$), low  error probability ($p_j$) and high AoI ($\delta^{rx}_j$). Formally, 
 \vspace{-0.1in}
\begin{align}
\pi(\delta^{rx}_1,\delta^{rx}_2,\ldots,\delta^{rx}_M) = \begin{cases}
\new_{\underset{j}\argmax(I_j(\delta^{rx}_j))} & \textrm{ if } \max_j I_j(\delta^{rx}_j) \geq \eta , \vspace{-0.1in} \\
\idle & \textrm{ otherwise.}
\end{cases}   
\end{align}


  The effectiveness of the WI policy is demonstrated in Section~\ref{sec:results}. The  WI policy, which corresponds to a {\highlight suboptimal} policy, can easily be shown to be optimal for our problem with standard ARQ if all the users are identical, i.e, $p_j=p$ and $w_j=w$, $\forall j$, and $\lambda=1$. 
  

\subsection{Lower Bound on the Average AoI under a Resource Constraint}
\label{sec:bound}

In this section, we derive a closed-form lower bound for the constrained MDP:

\begin{theorem}
For Problem~\ref{problem} with the standard ARQ protocol, we have $J_{LB}\leq J^{\pi}$, $\forall \pi \in \Pi$, where
\begin{align*}
J_{LB}=\frac{1}{2\lambda}{\left(\sum_{j=1}^M {\sqrt{\frac{w_j}{1-p_j}}}\right)}^2+\frac{\lambda w_{j^*}p_{j^*}}{2(1-p_{j^*})}+\frac{1}{2}\sum_{j=1}^M w_j, 
\textnormal{and }  j^*\triangleq \argmin_j{\frac{w_jp_j}{2(1-p_j)}} . \nonumber
\end{align*}
\label{thm:bound}
\end{theorem}
\vspace{-0.4in}
\begin{proof}
The proof is provided in Appendix~\ref{append_LB1}.
\end{proof}



Previously, \cite{Kadota2018} proposed a lower bound on the average AoI for a source node sending time-sensitive information
to multiple users through unreliable channels, without any resource constraint (i.e. $\lambda=1$). The lower bound in Theorem~\ref{thm:bound} shows the effect of the constraint $\lambda$, and even for $\lambda=1$, it is tighter than the one provided in \cite{Kadota2018}.

\section{AoI with Fixed Redundancy (FR) HARQ Protocol}
\label{sec:harq1}

In this section, the FR HARQ protocol, also studied in \cite{Najm2017, Yates2017}, is investigated for a multi-user status update system under a resource constraint. We assume that a generated status update contains $k_s$ information symbols and encoding of a status update is performed using an $(n_s, k_s)$-MDS code~\cite{Najm2017, Yates2017}. The transmission of a status update continues until $n_s$ symbols are transmitted either successfully or not. The receiver starts decoding after $n_s$ symbols are transmitted and the AoI drops to $n_s$ if at least $k_s$ transmissions are successful, otherwise increases by $n_s$.  In this case, each information symbol is considered as a packet which is transmitted in one time slot; that is, the minimum age which could be achieved is equal to $n_s$. 

We note that HARQ schemes other than FR HARQ can also be studied, e.g., chase combining when the base station retransmits the same packet and the receiver aggregates the energy from the repeated transmissions to increase signal to noise ratio (SNR), or incremental
redundancy (IR) HARQ which transmits additional redundancy bits in each retransmission and constantly adapts coding rate
until a successful decoding \cite{hybrid2001}. In this paper, a general HARQ model is studied and can be adapted to both chase combining and incremental redundancy. The particular HARQ protocol, i.e. FR HARQ with MDS coding, is chosen due to simplicity of computation and tractability of error probabilities. FR HARQ can be adopted to the general HARQ error probabilities with $r_{max}=n_s-1$:  $g(r)=1$ $\forall r=\{0,\ldots,n_s-2\}$ and $g(n_s-1)=p^{FR}_{j}\triangleq \textrm{Pr}(\textrm{less than } k_s \textrm{ symbols are received among } n_s \textrm{ transmissions})= \sum_{k=0}^{k_s-1}{{n_s}\choose{k}}p_j^{n_s-k}(1-p_j)^{k}$, $\forall j\in[M]$.

The problem for FR HARQ can be formulated as a RMAB problem, and whenever an arm (user) is chosen for transmission, a new update is generated and an encoded packet is transmitted for $n_s$ time slots to that user. If the idle action is chosen, the source stays idle for a single time slot.


{\highlight 
Similarly to Section~\ref{sec:arq}, low complexity heuristics based on the WI and a lower bound on the average AoI are presented for FR HARQ protocol. }


\begin{proposition}
Problem~\ref{problem} with the FR HARQ Protocol is indexable and the WI for each user can be computed in closed form. We have
\vspace{-0.1in}
\begin{align}
I^{FR}_j(\delta^{rx}_j)&= \frac{w_j}{2\left(\frac{n_s}{1-p^{FR}_{j}}\right)}\left(\left(\delta^{rx}_j+\frac{n_sp^{FR}_{j}}{1-p^{FR}_{j}}\right)^2+\delta^{rx}_j + \frac{n_sp^{FR}_{j}}{1-p^{FR}_{j}}-\frac{n_s^2p^{FR}_{j}}{(1-p^{FR}_{j})^2}\right), 
\end{align}
where
\begin{align}
p^{FR}_{j}&\triangleq \textrm{Pr}(\textrm{less than } k_s \textrm{ symbols are received})= \sum_{k=0}^{k_s-1}{{n_s}\choose{k}}p_j^{n_s-k}(1-p_j)^{k}. \label{eq:FR}
\end{align}
\label{proof_WI_HARQ}
\end{proposition}
\vspace{-0.45in}
\begin{proof}
The proof and the derivation of WI is given in Appendix~\ref{AppendixI}.
\end{proof}
{\highlight  Following the WI policy presented in Proposition 2, the source tends to transmit to a user more frequently as the age, the weight, and the error probability of the user increases.}

\begin{theorem}
For Problem~\ref{problem} with the FR HARQ Protocol, we have $J_{LB}\leq J^{\pi}$, $\forall \pi \in \Pi$, where
\begin{align*}
J_{LB}=\frac{n_s}{2\lambda}{\left(\sum_{j=1}^M {\sqrt{\frac{w_j}{1-p^{FR}_{j}}}}\right)}^2+\frac{\lambda n_sw_{j^*}p^{FR}_{j^*}}{2(1-p^{FR}_{j^*})}+\sum_{j=1}^M {w_j\left(n_s-\frac{1}{2}\right)}, \textnormal{and }  j^*\triangleq \argmin_j{\frac{w_jp^{FR}_j}{(1-p^{FR}_j)}}. \nonumber
\end{align*}
\label{thm_FIR}
\end{theorem}
\vspace{-0.45in}
\begin{proof}
The proof is provided in Appendix~\ref{AppendixJ}.
\end{proof}

 {\highlight  For any given network with ($w_j$, $p_j$, $\forall j$, and $\lambda$) and FR HARQ ($n_s,k_s$) protocol, the average AoI that can be obtained under any casual policy is higher than the closed-form lower bound provided in Theorem  4. The  expression in  Theorem 4 provides an intuition on  how  the  weights  ($w_j$), the error probabilities ($p_j$), the average transmission constraint ($\lambda$), the number of users ($M$) and the design of MDS coding ($n_s$,$k_s$) affect the performance of the system in terms of average AoI.}

{\highlight  Note that the results obtained for FR HARQ are identical to the ones obtained for standard ARQ protocol when ($n_s,k_s$) = $(1,1)$. If ($n_s,k_s$) is different than $(1,1)$, the average AoI result of Theorem 4 is equivalent to that of Theorem 3 scaled by $n_s$, where $p_j$s are replaced by $p^{FR}_j$ defined in \eqref{eq:FR}. }

\section{Learning in an unknown environment}
\label{sec:learning}

In sections~\ref{sec:solution}-\ref{sec:harq1}, it is assumed that the channel statistics change very slowly and the same transmission environment has been used for a long time before the time of deployment, i.e., the statistics regarding the error probabilities are available. In most practical wireless settings, however, the channel error probabilities for retransmissions may not be known at the time of deployment, or may change over time. We employ online learning algorithms to learn the error probabilities over time without degrading the performance significantly. In our previous work \cite{wcnc_paper,infocom_paper,journal_paper}, we proposed a simple \textit{average-cost SARSA} algorithm to minimize the average AoI for a single user. Due to the large state space of the multi-user network considered in this paper, different learning algorithms are considered.




\subsection{UCRL2 with HARQ}

The upper confidence RL (UCRL2) algorithm \cite{UCRL2} is a well-known RL algorithm for finite state and action MDP problems, with strong theoretical performance guarantees. However, the computational complexity of the algorithm scales quadratically with the size of the state space, which makes it unsuitable for large state spaces. UCRL2 has been initially proposed for generic MDPs with unknown rewards and transition probabilities; which need to be learned for each state-action pair. For the average AoI problem, the rewards are known (i.e., AoIs) while the transition probabilities are unknown. Moreover, the number of parameters to be learned can be reduced to the number of transmission error probabilities to each user; thus, the computational complexity can be reduced significantly. 

For a generic tabular MDP, UCRL2 keeps track of the possible MDP models (transition probabilities and expected immediate rewards) in a high-probability sense and finds a policy that has the best performance in the best possible MDP. To achieve this in our case, it is enough to optimistically estimate the error probabilities $g_j(r)$, and find a policy that is optimal for the resulting optimistic MDP. This is possible since the performance corresponding to a fixed sequence of transmission decisions improves if the error probabilities decrease. The average transmission constraint at the source requires additional modifications to UCRL2. We will guarantee this constraint by updating the Lagrange multiplier according to the empirical resource consumption. The details of the algorithm are given in Algorithm~\ref{alg:UCRL_HARQ}. 

UCRL2 exploits the optimistic MDP characterized by the optimistic estimation of error probabilities within a certain confidence interval, where $\hat{g}_j(r)$ and $\tilde{g}_j(r)$ represent the empirical and the optimistic estimates of the error probability for user $j$ after $r$ retransmissions.  In each episode, we keep track of a value $\eta$ resulting in a transmission cost close to $\lambda$, and then find and apply a policy that is optimal for the optimistic MDP (i.e., the MDP with the smallest total cost from among all plausible ones given the observations so far) with Lagrangian cost. 
In contrast to the original UCRL2 algorithm, finding the optimistic MDP in our case is easy (choosing lower estimates of the error probabilities), and we can use standard value iteration (VI) to compute the optimal policy (instead of the much more complex extended VI used in UCRL2). Thus, the computational complexity, which is the main drawback of UCRL2 algorithm, reduces significantly for the average AoI problem. UCRL2 is employed for Problem~\ref{problem} in this paper since it is an online algorithm (i.e., it does not need any previous training) and it enjoys strong theoretical guarantees for $\lambda =1$. The resulting algorithm will be called UCRL2-VI. 

\begin{algorithm}[h]
\begin{footnotesize}
\begin{algorithmic}[1]
 \renewcommand{\algorithmicrequire}{\textbf{Input:}}
 \renewcommand{\algorithmicensure}{\textbf{Output:}}
 \REQUIRE A confidence parameter $\rho\in (0,1)$, an update parameter $\alpha$, $\lambda$, confidence bound constant $U$, $|\mathcal{S}|$, $|\mathcal{A}|$
  \STATE $\eta=0$, $t=1$ and observe the initial state $s_1$. 
  \FOR {episodes $k= 1,2,\ldots$ } ~set $t_k\triangleq t$.
  \FOR {$j \in[M]$, $r\in [r_{max}]$}
  \STATE $N_k(j,r)\triangleq\vert\{\tau<t_k:a_{\tau}=\retx_j,r_{j,\tau}=r\}\vert$,
  ~ $N_k(j,0)\triangleq\vert\{\tau<t_k:a_{\tau}=\new_j\}\vert$.
  \STATE $E_k(j,r)\triangleq\vert\{\tau<t_k:a_{\tau}=\retx_j, r_{j,\tau}=r, NACK\}\vert$,
   $E_k(j,0)\triangleq\vert\{\tau<t_k:a_{\tau}=\new_j,  NACK\}\vert$.
  \STATE $\hat{g}_j(r)\triangleq \frac{E_k(j,r)}{\max\{N_k(j,r),1\}}$.
  \ENDFOR
  \STATE $C_k\triangleq\vert\{\tau<t_k:a_{\tau}\neq \idle\}\vert$.
  \STATE $\eta \leftarrow \eta+\alpha (C_k/t_k-\lambda)$. 
  \STATE Compute optimistic error probability estimates:
   $\qquad \tilde{g}_j(r)\triangleq\max\left\{0,\hat{g}_j(r)-\sqrt{\frac{U\log(|\mathcal{S}| |\mathcal{A}| t_k/\rho)}{max\{1,N_k(j,r)\}}}\right\}$.
  \STATE Use $\tilde{g}_j(r)$ and VI to find a policy $\tilde{\pi}_k$. 
  \STATE Set $v_k(j,r) \leftarrow 0$, $\forall j,r$.  
   \WHILE{$v_k(j,r)<N_k(j,r)$} \Comment{run policy $\tilde{\pi}_k$} 
\STATE Choose an action $a_t=\tilde{\pi}_k(s_t)$, and if $a_t \neq \idle$, set $j_t$ the target user, otherwise set $j_t=0$.
 \STATE Obtain cost $\sum_{j=1}^M w_j\delta^{rx}_j +\eta\mathbbm{1}[a_t\neq \idle]$ and observe $s_{t+1}$.
 \STATE Update $v_k(j_t,r)=v_k(j_t,r)+1$ and set $t\leftarrow t+1$.
  \ENDWHILE
  \ENDFOR
 \end{algorithmic} 
  \caption{UCRL2-VI}
  \label{alg:UCRL_HARQ}
  \end{footnotesize}
\end{algorithm}

\subsection{A Heuristic Version of the UCRL2 for Standard ARQ}

In this section, we consider the standard ARQ protocol with unknown error probabilities $p_j=g_j(0)$. The estimation procedure of UCRL2-VI can be immediately simplified accordingly, as it only needs to estimate $M$ parameters. In order to reduce the computational complexity, we can replace the costly VI in the algorithm to find the $\tilde{\pi}_k$ with the suboptimal WI  policy given in Section~\ref{sec:WI_arq}. The resulting algorithm, called \textit{UCRL2-Whittle}, selects policy $\tilde{\pi}_k$ in step 16 following the WI policy in Section~\ref{sec:arq}.  
The details of the algorithm are given in Algorithm \ref{alg:ARQ}, where  $\hat{p}(j)$ and $\tilde{p}(j)$ denote the empirical and the optimistic estimate of the error probability for user $j$. 

\begin{algorithm}[h]
\begin{footnotesize}
\caption{UCRL2 for the average AoI with ARQ.}
\begin{algorithmic}[1]
 \renewcommand{\algorithmicrequire}{\textbf{Input:}}
 \renewcommand{\algorithmicensure}{\textbf{Output:}}
 \REQUIRE A confidence parameter $\rho\in (0,1)$, an update parameter $\alpha$, $\lambda$, confidence bound constant $U$, $|\mathcal{S}|$,  $|\mathcal{A}|$
  \STATE $\eta=0$, $t=1$ and observe the initial state $s_1$.
  \FOR {episodes $k= 1,2,\ldots$ }
  \and set $t_k\triangleq t$, \STATE $N_k(j)\triangleq\vert\{\tau<t_k:a_{\tau}=\new_j\}\vert$, ~$E_k(j)\triangleq\vert\{\tau<t_k:a_{\tau}=\new_j, NACK\}\vert$
  \STATE $\hat{p}(j)\triangleq \frac{E_k(j)}{\max\{N_k(j),1\}}$,
   ~$C_k\triangleq\vert\{\tau<t_k:a_{\tau}\neq \idle\}\vert$,
  \STATE $\eta \leftarrow \eta+\alpha (C_k/t_k-\lambda)$. 
  \STATE Compute the optimistic error probabilities: $\tilde{p}(j)\triangleq\max\{0,\hat{p}(j)-\sqrt{\frac{U\log(|\mathcal{S}| |\mathcal{A}| t_k/\rho)}{max\{1,N_k(j)\}}}\}$
  \STATE Use $\tilde{p}(j)$  to find a policy $\tilde{\pi}_k$
  and execute policy $\tilde{\pi}_k$
  \WHILE{$v_k(j)<N_k(j)$}
 \STATE Choose an action $a_t=\tilde{\pi}_k(s_t)$,
 \STATE Obtain cost $\sum_{j=1}^M w_j\delta^{rx}_j +\eta* \mathbbm{1}[a_t\neq \idle]$ and observe $s_{t+1}$
 \STATE Update $v_k(j)=v_k(j)+1$, 
 \and set $t\leftarrow t+1$;
  \ENDWHILE
  \ENDFOR
 \end{algorithmic} 
  \label{alg:ARQ}
  \end{footnotesize}
 \end{algorithm}

\subsection{Average-Cost SARSA with LFA}
\label{sec:LFA}

In \cite{wcnc_paper}, the average-cost SARSA algorithm is employed with \emph{Boltzmann}  (\emph{softmax}) exploration for the average AoI problem with a single user. For the problem with multiple users, the cardinality of the state-action space is large and it is difficult to even store a matrix that has the size of the state-action space. Hence, average-cost SARSA with LFA is employed, where a linear function of features can be used to approximate the Q-function in SARSA~\cite{Puterman_book}. Average-cost SARSA with LFA is an online algorithm similar to average-cost SARSA and UCRL2 algorithms. It improves the performance of average-cost SARSA by improving the convergence rate significantly for multi-user systems and its application is much simpler than the UCRL2 algorithm. 

We approximate the $Q$ function with a linear function $Q_{\theta}$ defined as: $Q_{\theta}(s,a) \triangleq \theta^T \phi(s,a)$, where $\phi(s,a)\triangleq(\phi_1(s,a),\ldots,\phi_d(s,a))^T$ is a given feature associated with the pair $(s,a)$. In our experiments, we set  $\{\phi_i(s,a)\}_{i=1}^{M}$ as the weighted age at the receiver of each user ($w_j\delta^{rx}_j$), $\{\phi_i(s,a)\}_{i=M+1}^{2M}$ as the age at the transmitter of each user ($\delta^{tx}_j$) and $\{\phi_i(s,a)\}_{i=2M+1}^{3M}$ as the retransmission number of each user ($r_j$) given an action $a\in \mathcal{A}$ is chosen in state $s\in \mathcal{S}$: 
\begin{align}
Q_{\theta}(s,a)=\theta_{(0,a)}+\theta_{(1,a)} w_1\delta^{rx}_1+\ldots+\theta_{(M,a)}w_M\delta^{rx}_M+\theta_{(M+1,a)} w_1\delta^{rx}_1+\ldots\nonumber\\+\theta_{(2M,a)}w_M\delta^{rx}_M+\theta_{(2M+1,a)} r_1+\ldots+\theta_{(3M,a)}r_M,
\end{align}
where $\theta_{(0,a)}$ denotes the constant variable. The dimension of $\theta$ is $d=(3M+1)|\mathcal{A}|.$ The outline of the algorithm is given in Algorithm \ref{algo_SARSA}.

\begin{algorithm}
\begin{footnotesize}
\caption{Average-cost SARSA with LFA}
\begin{algorithmic}[1]
 \renewcommand{\algorithmicrequire}{\textbf{Input:}}
 \renewcommand{\algorithmicensure}{\textbf{Output:}}
\REQUIRE Lagrange parameter $\eta$, update parameters $\alpha$, $\beta$, $\gamma$, $\mathcal{A}$, and set $t\leftarrow 1~$, $\theta \leftarrow 0$, $J_{\eta}\leftarrow 0~$   
\FOR{$t=1,2,\ldots$}    
  
\STATE   Find the parameterized policies with Boltzmann exploration: $\pi(a|s_t)=\frac{\exp(-\theta^T\phi(s_t,a) )}{\sum_{a'\in\mathcal{A}}{\exp(-\theta^T\phi(s_t,a') )}}$. 
    \STATE  Sample and execute action $a_t$ from $\pi(a|s_t)$.
    \STATE Observe the next state $s_{t+1}$ and cost $\sum_{j=1}^M \delta^{rx}_j +\eta* \mathbbm{1}[a_t\neq \idle]$.
    \STATE $\pi(a|s_{t+1})=\frac{\exp(-\theta^T\phi(s_{t+1},a))}{\sum_{a'_{t+1}\in\mathcal{A}}{\exp(-\theta^T\phi(s_{t+1},a'))}}$ 
	\STATE Sample $a_{t+1}$ from $\pi(a|s_{t+1})$ \;
    \STATE Compute $C_{\eta}$
    \STATE Update linear coefficients: $ \theta \leftarrow \theta + \alpha_t [\Delta+\eta\cdot\mathbbm{1}[a_t\neq \idle]-J_{\eta}+\theta^T\phi(s_{t+1},a_{t+1})-\theta^T\phi(s_t,a_t)] \phi(s_t,a_t)$,
     \STATE   Update gain: $J_{\eta}\leftarrow J_{\eta}+ \beta_t [\Delta+\eta\cdot\mathbbm{1}[a_t\neq \idle]-J_{\eta}]$, 
     \STATE Update Lagrange multiplier: $\eta \leftarrow \eta + \gamma_t (C_{\eta}-\lambda)$

\ENDFOR
\end{algorithmic} 
\label{algo_SARSA}
\end{footnotesize}
\end{algorithm}

The performance of average cost SARSA with LFA is demonstrated in Section~\ref{sec:results}. We note that linear approximators are not always effective, and the performance can be improved in general by using a non-linear approximator. However; the performance also depends on the availability of data, i.e., the linear approximator may perform better if the available data set is limited.)


\subsection{Deep Q-Network (DQN)}

A DQN uses a multi-layered neural network in order to estimate the values of $Q(s,a)$; that is, for a given state $s$, DQN outputs a vector of state-action values, $Q_{\theta}(s, a)$, where $\theta$ denotes the parameters of the network. That is, the neural network is a function from $2M$ inputs to $|\mathcal{A}|$ outputs which are the estimates of the Q-function $Q_{\theta}(s,a)$.  We apply the DQN algorithm of \cite{Mnih2015} to learn a scheduling policy. We create a fairly simple feed-forward neural network of $3$ layers, one of which is the hidden layer with 24 neurons. We also use \textit{Huber loss} \cite{huber1964} and the \textit{Adam} algorithm \cite{Adam2014} to conduct stochastic gradient descent to update the weights of the neural network. 

We exploit two important features of DQNs  as proposed in \cite{Mnih2015}: \textit{experience replay} and a \textit{fixed target network}, both of which provide algorithm stability. For \textit{experience replay}, instead of training the neural network with a single observation $<\!s,a,s',c(s,a)\!>$ at the end of each step,  many experiences (i.e., (state, action, next state, cost) quadruplets) can be stored in the replay memory for batch training, and a minibatch of observations randomly sampled at each step can be used. The DQN uses two neural networks: a target network and an online network. The \textit{target network}, with parameters $\theta^-$, is the same as the online network except that its parameters are updated with the parameters $\theta$ of the online network after every $T$ steps, and $\theta^-$ is kept fixed in other time slots. For a minibatch of of observations for training, temporal difference estimation error $e$ for a single observation can be calculated as 
\vspace{-0.1in}
\begin{align}
    e = Q_{\theta}(s,a)-(-c(s,a)+\gamma Q_{\theta^-}(s',\argmax Q_{\theta}(s',a))).
\end{align} 

\textit{Huber loss} is defined by the squared error term for small estimation errors, and a linear error term for high estimation errors, allowing less dramatic changes in the value functions and further improving the stability. For a given estimation error $e$ and loss parameter $d$, the Huber loss function, denoted by $\mathrm{L}^d(e)$, and the average loss over the minibatch, denoted by $\mathcal{B}$, are computed as
\begin{align*}
\mathrm{L}^d(e)=
    \begin{cases}
    e^2 &\textrm { if } e\leq d \\
    d (|e|-\frac{1}{2}d)) &\textrm { if } e> d,
    \end{cases}
     ~~\textrm{and }~~\mathrm{L}_{\mathcal{B}}= \frac{1}{|\mathcal{B}|}\sum_{<s,a,s',c(s,a)>\in \mathcal{B}} \mathrm{L}^d(e).
\end{align*}

We apply the $\epsilon$-greedy policy to balance
 exploration and exploitation, i.e., with probability $\epsilon$ the source randomly selects an action, and with probability $1-\epsilon$ it chooses the action with the minimum Q value.  We let $\epsilon$ decay gradually from $\epsilon_0$ to $\epsilon_{min}$; in other words, the source explores more at the beginning of training and exploits more at the end. The hyperparameters of the DQN algorithm are  tuned for our problem experimentally, and are given in Table~\ref{table_DQN}.
\begin{table}
\caption{Hyperparameters of DQN algorithm used in the paper}
\vspace{-0.3in}
\begin{footnotesize}
\begin{center}
 \label{table_DQN}
 \scriptsize
\begin{tabular}{ |c|c|c|c|c|c|c|c| } 
\hline
 Parameter & Value & Parameter & Value & Parameter & Value & Parameter & Value \\ 
 \hline \hline
 discount factor $\gamma$ & 0.99 & optimizer & Adam & activation function & ReLU &learning rate $\alpha$  & $10^{-4}$\\ 
 \hline
 minibatch size & 32 & loss function & Huber loss & hidden size & 24  & $\epsilon$ decay rate $\beta$ & 0.9\\ 
 \hline
 replay memory length & 2000 & exploration coefficient $\epsilon_0$ & 1 & episode length $T$ & 1000  & $\epsilon_{min}$ & 0.01 \\ \hline 
\end{tabular}
\end{center}
\end{footnotesize}
\end{table}

\section{Numerical Results}
\label{sec:results}

\begin{table}
 \caption{{\highlight A summary of RL algorithms presented in this paper}}
 \vspace{-0.3in}
 \begin{footnotesize}
\begin{center}
 \label{table_RL}
  \scriptsize
 \begin{tabular}{ |p{0.2\textwidth}|p{0.35\textwidth}|p{0.35\textwidth}|} 
 \hline
  RL Method & Advantages & Disadvantages \\ 
  \hline \hline
 RVI~\cite{Puterman_book} & simple, converges to optimal for MDPs & requires apriori information on system characteristics \\ \hline
 Average cost SARSA (tabular)~\cite{journal_paper} & simple, fully online & does not perform well for large state spaces, requires an approximation to finite state spaces \\ 
  \hline
     Average cost SARSA with LFA  & converges faster than Average cost SARSA  applicable to infinite state spaces & convergences slower than UCRL2 and DQN, stability issues for average cost problems\\ 
  \hline
  UCRL2-VI & theoretical convergence guarantee & large computational complexity due to VI \\ 
  \hline
  UCRL2-Whittle & low computational complexity & based on the computation of WI   \\ \hline 
  DQN~\cite{Mnih2015} & performs well and applicable to infinite or large state spaces & requires pre-training \\ \hline 
 \end{tabular}
 \end{center}
 \end{footnotesize}
\end{table}

In this section, we provide numerical results for the proposed learning algorithms, and compare the achieved average performances.  First, we analyze the average AoI with the standard ARQ protocol. The asymptotic average AoI as a function of the resource constraint $\lambda$ is shown in Figure~\ref{fig:cons_vs_AoI} for a 3-user system with error probabilities $p=g(0)=[0.5~ 0.2~ 0.1]$.  It can be seen from Figure~\ref{fig:cons_vs_AoI} that both UCRL2-VI and UCRL2-Whittle perform very close to the lower bound, particularly when $\lambda$ is small, i.e., the system is more constrained. Although UCRL2-Whittle has a significantly lower computational complexity, it performs very close to UCRL2-VI for all $\lambda$ values. 

\begin{figure}[!t]
\centering
\includegraphics[scale=0.5]{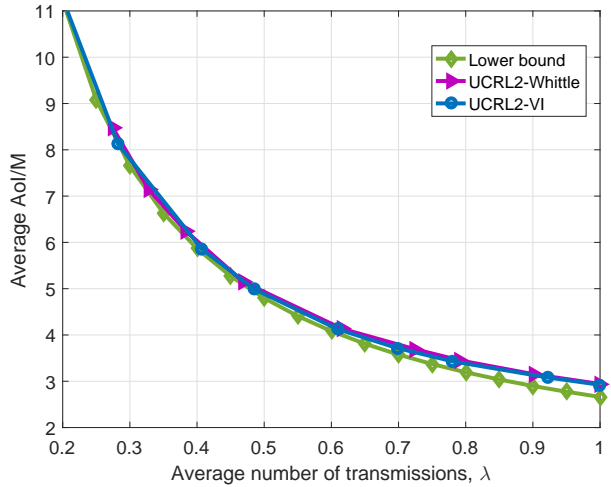}
\vspace{-0.2in}
\caption{Average AoI with respect to $\lambda$ for a 3-user network under the standard ARQ protocol, with error probabilities $p=[0.5~ 0.2~ 0.1]$, and $w_j=1,~\forall j$. Time horizon is set to $T=10^5$, and the results are averaged over $100$ runs.}
\label{fig:cons_vs_AoI}
\end{figure}

Figure~\ref{fig:arq_size} illustrates the mean and variance of the average AoI with standard ARQ with respect to the size of the network when there is no constraint on the average number of transmissions (i.e. $\lambda=1$) and the performance of the UCRL2-Whittle is compared with the lower bound (UCRL2-VI is omitted since its performance is very similar to UCLR2-Whittle and has a much higher computational complexity, especially for large $M$). The performance of UCRL2-Whittle is close to the lower bound and is very similar to that of the WI policy,  which requires a priori knowledge of the error probabilities. {\highlight  Moreover, our algorithm outperforms the benchmark \textit{greedy policy}, which always transmits to the user with the highest age (i.e., $a=\new_j$, such that $j=argmax \delta^{rx}_j, \forall j\in [M]$), as well as the \textit{round robin policy}, which transmits to each user in turns.}

\begin{figure}[!t]
\centering
\includegraphics[scale=0.5]{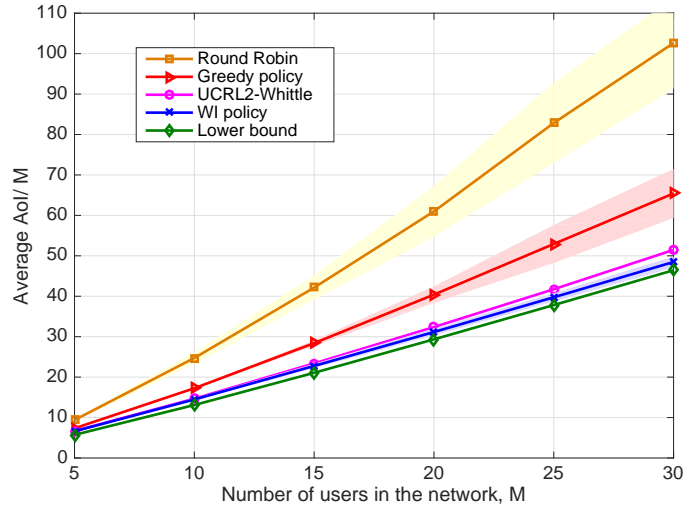}
\vspace{-0.2in}
\caption{Average AoI under the standard ARQ protocol for networks with different sizes, where $p_j=(j-1)/M$, $\lambda=1$, and  $w_j=1,~\forall j$. The results are obtained after $10^4$ time steps and averaged over 100 runs (both the mean and the variance are shown).}
\label{fig:arq_size}
\end{figure}

The performance of the proposed RL algorithms (the mean and the variance of the average AoI) is compared with the performance of average cost SARSA, proposed in \cite{wcnc_paper} for a point-to-point status update system (which is used as a benchmark policy in this paper), in Figure~\ref{fig:learn_ARQ}. The DQN algorithm in the figure is configured as in Table~\ref{table_DQN} and trained for $500$ episodes. The average AoI for DQN is obtained after $10^5$ time steps and averaged over $100$ runs. UCRL2-Whittle and average cost SARSA with LFA converge much faster compared to the standard average-cost SARSA, and they perform very close to the transmission scheduling computed by RVI with known error probabilities. Although DQN and UCRL2-Whittle perform better than average cost SARSA with LFA, DQN requires a training time before running the simulation.


Figure~\ref{fig:learn_HARQ} shows the performance of the learning algorithms for the HARQ protocol (the mean and the variance of the average AoI) for a 2-user scenario. Similarly to Figure~\ref{fig:learn_ARQ}, DQN is trained for 500 episodes with configuration in Table~\ref{table_DQN}.  It is worth noting that although UCRL2-VI converges to the optimal policy in fewer iterations than average-cost SARSA and average-cost SARSA with LFA, iterations in  UCRL2-VI are computationally more demanding since the algorithm uses VI in each epoch. Therefore, UCRL2-VI  is not practical for problems with large state spaces, in our case for large $M$. On the other hand, UCRL2-Whittle can handle a large number of users since it is based on a simple index policy instead of VI. As illustrated in both Figures~\ref{fig:learn_ARQ} and \ref{fig:learn_HARQ} that LFA significantly improves the performance of average cost SARSA and DQN with neural network estimator, and UCRL2-Whittle improves the performance of RL even more.



\begin{figure}[!t]
\centering
\includegraphics[scale=0.5]{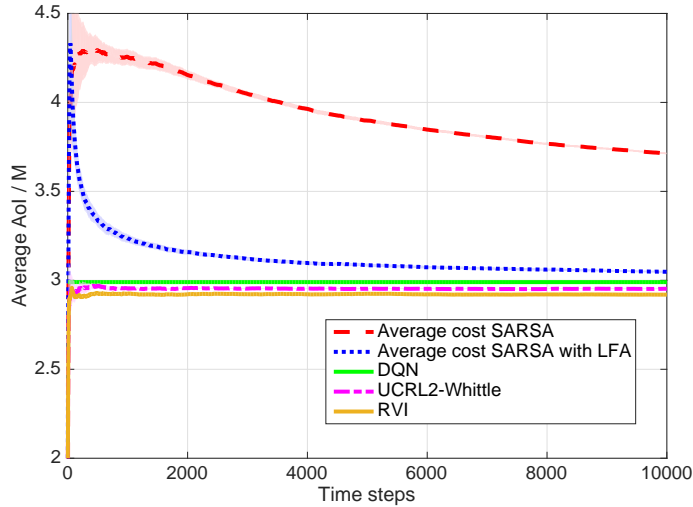}
\vspace{-0.2in}
\caption{Average AoI for a 3-user ARQ network with error probabilities $p=[0.5~ 0.2~ 0.1]$, $\lambda=1$, and $w_j=1,~\forall j$. The results are obtained after $10^4$ time steps and averaged over 100 runs (both the mean and the variance are shown).}
\label{fig:learn_ARQ}
\end{figure}

\begin{figure}[!t]
\centering
\includegraphics[scale=0.5]{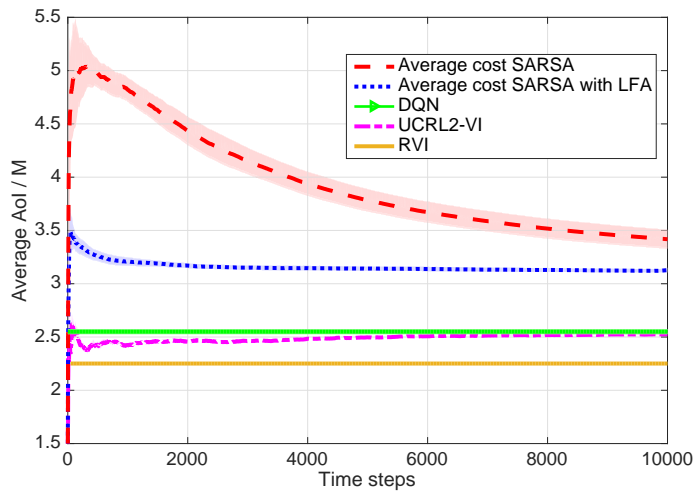}
\vspace{-0.2in}
\caption{Average AoI for a 2-user HARQ network with error probabilities $g_1(r_1)=0.5 \cdot 2^{-r_1}$ and $g_2(r_2)=0.2 \cdot 2^{-r_2}$, where $\lambda=1$ and $w_j=1,~\forall j$. The simulation results are averaged over 100 runs (both the mean and the variance are shown).}
\label{fig:learn_HARQ}
\end{figure}

The performance of FR HARQ protocol as described in Section~\ref{sec:harq1} for packets MDS-coded with $(n_s,k_s)=(5,3)$ is shown in Figure~\ref{fig:harq_size}. The probability that an MDS-coded packet is not correctly decoded is given in \eqref{eq:FR} where the symbol transmission error probability is set to $p_{s,j}=(j-1)/2M$ for user $j$. As Figure~\ref{fig:harq_size} illustrates, average AoI per user increases linearly with the number of users in the network and the WI policy performs very close to the lower bound for both $\lambda=0.6$ and $\lambda=1$.

\begin{figure}[!t]
\centering
\includegraphics[scale=0.5]{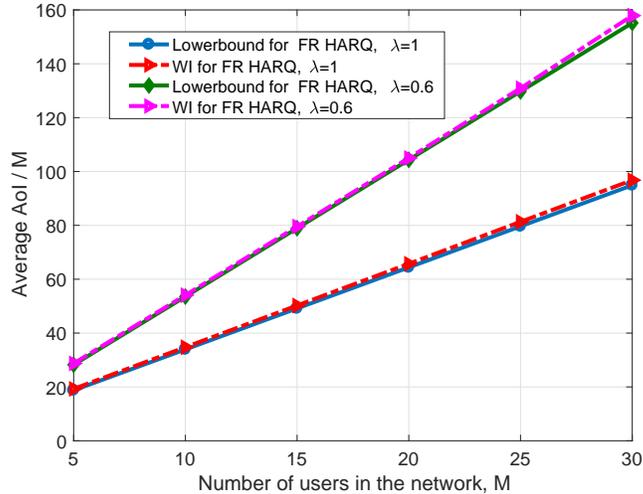}
\vspace{-0.2in}
\caption{Average AoI for networks with different sizes and constraints, where $p_{s,j}=(j-1)/2M$, $(n_s,k_s)=(5,3)$, and  $w_j=1,~\forall j$. The results are obtained after $10^4$ time steps and averaged over 100 runs.}
\label{fig:harq_size}
\end{figure}

Figure~\ref{fig:DQN_HARQ} shows the evolution of average AoI across 10 users with DQN after different number of training episodes, where each training episode consists of $1000$ time steps.  Following \cite{IEEEstandard,LTEmax}, $r_{max}$ is set to 3 and general HARQ protocol is considered with $g_j(r_j)=(j-1)/M \cdot 2^{-r_j}$ motivated by the exponentially decreasing error profile of HARQ protocols studied in \cite{harq2003,hybrid2001}. Figure~\ref{fig:DQN_HARQ} illustrates that average AoI achieves its minimum after about 150 episodes of training. We note that we did not run UCRL2-VI and average-cost SARSA algorithms for a 10-user HARQ problem since state space is large and convergence takes too long compared to the DQN algorithm.

\begin{figure}[!t]
\centering
\includegraphics[scale=0.5]{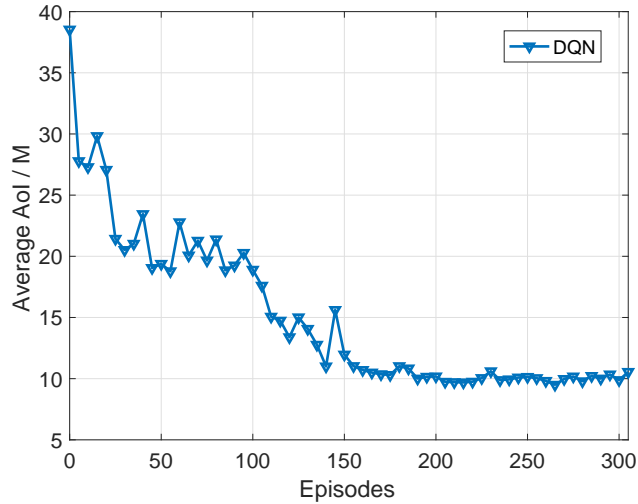}
\vspace{-0.2in}
\caption{Average AoI obtained by DQN algorithm with respect to the number of training episodes for a 10-user HARQ network with error probabilities $g_j(r_j)=(j-1)/M \cdot 2^{-r_j}$, $j \in [M]$ and $r_j \in [r_{max}]$, $\lambda=1$, $r_{max}=3$ and $w_j=1,~\forall j$. Each episode consists of $1000$ time steps and the results are obtained for a single run.}
\label{fig:DQN_HARQ}
\end{figure}

{\highlight  A summary of the employed RL algorithms together with their strengths and weakness are given in Table~\ref{table_RL}.} We concluded that the  choice of the learning algorithm to be adopted depends on the scenario and system characteristics. It has been shown that average-cost SARSA is not effective considering the large state space of the multi-user problem. Different state-of-the-art RL methods are presented including SARSA with LFA, UCRL2, and DQN. The performance of UCRL2-VI algorithm is close to optimal for small networks, i.e. consisting of 1-5 users, and enjoys theoretical guaranties. However, UCRL2-VI is not favorable for large networks due to its computational complexity resulting from value iteration, and UCRL2-Whittle is preferable. On the other hand, UCRL2-Whittle cannot be employed for a general HARQ multi-user system. Similarly, SARSA with LFA has decreased the average AoI significantly for small-size networks with HARQ; however, is not effective for large networks and SARSA with LFA lacks stability. A non-linear approximation with DQN performs well for large networks, while it is not fully online and it requires a training time before running the algorithm.

\section{Conclusion}
\label{sec:conclusion}

We considered scheduling the transmission of status updates to multiple users with the weighted average AoI as the performance measure. Under a resource constraint at the source node, the problem is modeled as a CMDP and the structure of the optimal policy is established. Lower bounds on the average AoI are derived for special cases of the problem. RL algorithms were presented for scenarios where the error probabilities may not be known in advance, and were numerically shown to provide near-optimal performance in simple scenarios. It has been demonstrated that the optimal choice of the learning algorithm to be adopted depends on the scenario and system characteristics.  The algorithms presented in this paper are also relevant to other multi-user systems concerning the timeliness of information. {\highlight The AoI for multi-user systems without feedback, or under imperfect and delayed ACK/NACK feedback will be studied as a future work. }

\appendix

\subsection{Proof of Proposition~\ref{prop_WI}} 
\label{append_WI}


We note that each sub-problem \eqref{eq:BellmanC}  coincides with the Lagrangian formulation of the single-user problem which we previously studied in \cite{wcnc_paper,journal_paper}. According to Lemma~1 of~\cite{journal_paper}, the policy which solves the Bellman optimality equations is a threshold policy, such that $a=\new_j$ if and only if $\delta^{rx}_j\ge\Gamma_j$
for an appropriate threshold $\Gamma_j$. As a consequence, the cost function (the sum of the average AoI and the average transmission cost for that user), given a threshold $\Gamma_j$ and cost of transmission $C$, can also be obtained in closed form as
\begin{align}
\label{eq:JetaARQ_m}
L^{\Gamma_j}_{C}=\frac{1}{\Gamma_j+\frac{p_j}{1-p_j}} \left(\frac{(\Gamma_j-1)\Gamma_j}{2}+ \frac{C+\Gamma_j}{1-p_j}+\frac{p_j}{(1-p_j)^2}\right). 
\end{align}

Using Definition~\ref{def:whittle} and \eqref{eq:JetaARQ_m}, we can compute the WI in closed form: By the definition of the threshold policy, we can find a $C$ such that both choices of thresholds $\Gamma_j=\delta^{rx}_j$ and $\Gamma_j=\delta^{rx}_j+1$ result in the same average cost, i.e., the average cost $L_C^{\delta^{rx}_j}$ should be equal to $L_C^{\delta^{rx}_j+1}$, which can be computed using \eqref{eq:JetaARQ_m}:
\begin{equation}
I_j(\delta^{rx}_j)= \frac{1}{2} w_j \delta^{rx}_j (1-p_j)(\delta^{rx}_j+\frac{1+p_j}{1-p_j}).
\end{equation}

The ${(M+1)}^{th}$ arm stands for the idle action. The Lagrange multiplier $\eta$ represents the cost of transmission and $C$ represents the cost of staying idle. If $\eta^*$ is equal to $C$ then both actions are equally desirable; that is, the WI for the ${M+1}^{th}$ arm is $I_{M+1}=\eta^*$. 


Note that an optimal threshold $\Gamma_j^*$ for a given $C$, which minimizes \eqref{eq:JetaARQ_m}, can be computed for a given $C$ as follows:
\begin{align*}
\Gamma^*_j \in \left\{ \left\lfloor\frac{\sqrt{2C(1-p_j)+p_j}-p_j}{1-p_j} \right\rfloor, \left\lceil\frac{\sqrt{2C(1-p_j)+p_j}-p_j}{1-p_j} \right\rceil  \right\}.
\end{align*}
As C increases from $0$ to $\infty$, $\Gamma_j^*$ monotonically increases from $0$ to $\infty$, and $S_j^{\new_j}(C)$ monotonically decreases from the entire state space $\mathcal{S}$ to an empty set. Thus, the problem is indexable. 
\qed

\subsection{Proof of Theorem~\ref{thm:bound}} 
\label{append_LB1}

 \begin{proof}

The system model and the definition of action $a_t\in \mathcal{A}$ implies the following constraints in addition to the average number of transmissions constraint in Problem~\ref{problem}: (i) updates occur in discrete time slots, and (ii) collisions are not allowed, i.e., no more than one user can be updated in a slot. In order to derive the lower bound, we relax the constraints (i) and (ii). 

First, we relax constraint (ii) and decouple the model to $M$ point-to-point status update systems each with a single user to serve. Let $J^{\pi}_j$ and $C_j^{\pi}$ denote respectively, the expected average AoI and the expected average number of transmissions for user $j$ if we follow policy $\pi$.  Assume that each user $j$ has an average number of transmissions constraint of $\lambda_j$ is imposed on user $j$, and we have:
\begin{align}
J^{*} \geq \sum_{j=1}^M w_jJ^*_j \geq \sum_{j=1}^M w_jJ_{j,LB}, ~\textrm{ given that } \sum_j^{M}{\lambda_j}=\lambda, \label{eq:decoupling}
\end{align}
where $J^*_j\leq J^{\pi}_j$, $\forall\pi$, denotes the minimum expected average AoI for user $j$ given $C^{\pi}_j \leq \lambda_{j}$ and $J_{j,LB}$ denotes a lower bound on the average AoI for user $j$.

The first inequality in \eqref{eq:decoupling} results from the relaxation of (ii) and decoupling the users. Then, we minimize the average AoI for a single user $j$ under a constraint $C_j\leq \lambda_j$, which reduces Problem~\ref{problem} to a single user problem which we previously studied in Section~\ref{sec:arq} of \cite{journal_paper}. 
The second inequality in \eqref{eq:decoupling} is due to relaxing the discrete time assumption in (i) in order to find a lower bound $J_{j,LB}$ on the average AoI for user $j$ by using closed form average AoI and resource consumption expressions also obtained in \cite{journal_paper}.


According to Theorem~2 of~\cite{journal_paper}, the policy that solves the Bellman optimality equations is a threshold policy. Then, for a given threshold, the expected average AoI and average number of transmissions can be computed in closed form similarly to~\cite{journal_paper}:
\begin{align}
 J^{\Gamma_j}_j = \frac{(\Gamma_j(1-p_j)+p_j)^2+p_j}{2(1-p_j)(\Gamma_j(1-p_j)+p_j)}+\frac{1}{2}, \textrm{ and }
C^{\Gamma_j}_j = \frac{1}{\Gamma_j (1-p_j)+p_j}. \label{eq:cj}
\end{align}

The lower bound $J_{j,LB}$ on the average AoI for a single user can be computed by substituting $C^{\Gamma_j}_j$ into $J^{\Gamma_j}_j$ and using the constraint $C^{\Gamma_j}_j\leq \lambda_j$: 
\begin{align}
    J^{\Gamma_j}_j = \frac{1/C_j^2+p_j}{2(1-p_j)/C_j}+\frac{1}{2} \geq \frac{1/\lambda_j^2+p_j}{2(1-p_j)/\lambda_j}+\frac{1}{2} \triangleq J_{j,LB}, \label{eq:single_lower}
\end{align}
and so $J^{*}_j\geq J_{j,LB}$ since $J^{\Gamma_j}_j\geq J_{j,LB}$ for all $\Gamma_j$ values satisfying $C^{\Gamma_j}_j\leq \lambda_j$. 
By inserting \eqref{eq:single_lower} to \eqref{eq:decoupling}, we obtain
\begin{align}
J^{*} \geq \sum_{j=1}^M {w_j\left(\frac{1/\lambda_j^2+p_j}{2(1-p_j)/\lambda_j}+\frac{1}{2}\right)} ~\textrm{ given that } \sum_j^{M}{\lambda_j}=\lambda. \label{eq:sum}
\end{align}
Note that the right hand side (RHS) of \eqref{eq:sum} is a convex function of $\lambda_1,\ldots,\lambda_M$. Then, the optimal $\lambda_j$ to minimize the RHS of \eqref{eq:sum} can be found numerically when $p_j$ and $\lambda$ are given. However, in order to obtain a closed form solution which can be easily computed and compared to the state-of-art bounds in the literature,  we approximate the RHS and obtain a slightly looser bound on the performance in closed form:
\begin{subequations}
\begin{align}
J^{*} &\geq  \min_{\lambda_1,\ldots,\lambda_M: \sum_{i=1}^M{\lambda_i= \lambda}}{\sum_{j=1}^M {\frac{w_j}{2\lambda_j(1-p_j)}}}+\min_{\lambda_1,\ldots,\lambda_M: \sum_{i=1}^M{\lambda_i= \lambda}}{\sum_{j=1}^M {\frac{w_j\lambda_jp_j}{2(1-p_j)} }}   +\frac{1}{2}\sum_{j=1}^M {w_j} \label{eq:3}\\
&=\frac{1}{2\lambda}{\left(\sum_{j=1}^M {\sqrt{\frac{w_j}{1-p_j}}}\right)}^2+\frac{\lambda}{2}\min_j{ \left(\frac{w_jp_j}{1-p_j}  \right)   }+\frac{1}{2}\sum_{j=1}^M w_j. \label{upperbound}
\end{align}
\end{subequations}

Here inequality \eqref{eq:3} results from the fact that independently minimizing the terms of a sum is smaller than the minimization of the sum; and the first term in \eqref{upperbound} is equal to the minimum of the first term in \eqref{eq:3} with the constraint of  $\sum_{j=1}^M{\lambda_j}=\lambda$, where a solution for $\lambda_j$ is found using the Lagrangian relaxation and the Karush–Kuhn–Tucker conditions \cite{Boyd2004}, leading to
\begin{align}
\lambda^*_j=\frac{\sqrt{w_j/(1-p_j)}}{\sum_{i=1}^M{\sqrt{w_i/(1-p_i)}}}.
\end{align}
The second term in \eqref{upperbound} is the minimum of the second term in \eqref{eq:3} under the constraint of $\sum_{j=1}^M{\lambda_j}=\lambda$,  which proves the theorem. 
\end{proof}

\subsection{Proof of Proposition~\ref{proof_WI_HARQ}}

\label{AppendixI} 




Following similar steps to Appendix~\ref{append_WI}, the problem can be approximated by decoupling the system into $M$ point-to-point status update systems with FR HARQ, where the cost of a single transmission is $C$. The state-action cost function and optimality equations are given.
\begin{align}
 h_c(\delta,0)&=\min\left( Q(\delta,0,\new_j), Q(\delta,0,\idle)\right), \\
    Q(\delta,0,\new_j) &=\left(\sum_{n=0}^{n_s-1}{\delta+n+C-L_j^C}\right)+p_eh_C(\delta+n_s,0)+(1-p_e)h_C(n_s,0),\label{eq:new_action}\\
    Q(\delta,0,\idle)&=\left(\sum_{n=0}^{n_s-1}{\delta+n-L_j^C}\right)+h_c(\delta+n_s,0). \label{eq:idle_action}
\end{align}


Next we investigate the optimal policy in the subsystem for a single user (user $j$), treated independently from the other decisions, which can be shown to be of threshold type:

\begin{lemma}
The decision to start transmitting to user $j$ ($a_j=\new_j$) is monotone with respect to the age $\delta^{rx}_j$, that is if $a_j^*(\delta^1,0)=\new_j$, then $a_j^*(\delta^2,0)=\new_j$ for all $\delta^2 \ge \delta^1$. \end{lemma}
\vspace{-0.2in}
\begin{proof}
A monotone threshold policy is optimal if $Q(\delta^{rx}_j,0,a_j)$ has a \textit{sub-modular} structure  in $(\delta^{rx}_j,a_j)$\cite{Topkis1978}, that is, 
\vspace{-0.2in}
\begin{align}
Q(\delta^1,\new_j)-Q(\delta^1,\idle)\geq Q(\delta^2,\new_j)-Q(\delta^2,\idle), \label{eq:submodular_m}
\end{align} 
for any $\delta^2\geq \delta^1$. From \eqref{eq:idle_action} and \eqref{eq:new_action}, for any $\delta>0$, we have
\begin{align}
Q(\delta,0,\new_j)-Q(\delta,0,\idle)&=n_sC+(1-p_e) h_{C}(n_s,0)-(1-p_e) h_{C}(\delta+n_s).
\end{align}
We can see that (\ref{eq:submodular_m}) holds if and only if $h_{C}(\delta,0)$ is a non-decreasing function of the age. We compare the costs incurred by the systems starting in states $\delta^1$ and $\delta^2$ via coupling the stochastic processes governing the behavior of the system; that is, we assume that the realization of the channel behavior is the same for both systems over the time horizon (this is valid since channel states/errors are independent of the ages and the actions). Assume a sequence of actions $\{a^2_t\}_{t=1}^{\infty}$ corresponds to the optimal policy starting from age $\delta^2$ for a particular realization of channel errors, and let $\{\delta^i_t\}$ denote the sequence of states obtained after following actions $\{a^2_t\}$ starting from state $\delta^{rx}_1=\delta^i$, $i=1,2$. Then, if $\delta^1 \le \delta^2$, clearly $\delta^1_t \le {\delta^{rx}_2}^t$ for all $t$.
Furthermore, by the Bellman optimality equation \eqref{eq:Bellman},
\begin{align*}
h_{C}(\delta^1,0) &  \le \Exp{\sum_{t=1}^\infty (\delta^1_t+ C \cdot \mathbbm{1}[a^2_t \neq \idle]-L^*_j) \bigg| \delta^1_1=\delta^1} \\
& \le \Exp{\sum_{t=1}^\infty (\delta^2_t+ C \cdot \mathbbm{1}[a^2_t \neq \idle]-L^*_j) \bigg| \delta^1_1=\delta^2} = h_{C}(\delta^2,0)~.
\end{align*}
This completes the proof of the lemma.
\end{proof}


Note that under a threshold policy with threshold $\Gamma_j$, the AoI process is a renewal process with i.i.d. renewal periods of $X\triangleq\Gamma_j-n_s+S_j$, where $S_j$ is the random time between the start of transmission of a status update and successful decoding of that update at the receiver of user $j$ similarly to \cite{Najm2017,Yates2017,Najm2019}. Then, the average AoI can be be written as expectation of the area under the AoI graph divided by the expected value of $X$ (denoted by $\Exp{X}$), which is given by:
\vspace{-0.1in}
\begin{align}
J_j^{\Gamma_j} = \frac{\Exp{S^2}+(\Gamma_j-n_s)\Exp{S_j}}{2((\Gamma_j-n_s)+\Exp{S_j})} +\frac{\Gamma_j+n_s}{2}-\frac{1}{2},
\label{eq:FR_age}
\end{align}
where the constant ($-1/2$) results from the fact that we consider a stair-step function to represent the AoI. Similarly, the average number transmissions is given by:
\begin{align}
C_j^{\Gamma_j} = \frac{\Exp{S_j}}{\Exp{X}}= \frac{\Exp{S_j}}{(\Gamma_j-n_s)+\Exp{S_j}}.
\label{eq:FR_transmission}
\end{align}
Thus, we have
\begin{align}
L_C^{\Gamma_j} = \frac{\Exp{S_j^2}+(\Gamma_j-n_s)\Exp{S_j}}{2((\Gamma_j-n_s)+\Exp{S_j})} +\frac{\Gamma_j+n_s}{2}-\frac{1}{2}+ C\cdot \frac{\Exp{S_j}}{(\Gamma_j-n_s)+\Exp{S_j}}, \label{eq:FR_whole}
\end{align}
where $\Exp{S_j}=\frac{n_s}{1-p^{FR}_j}$, and $\Exp{S_j^2}=\frac{n_s^2 (1+p^{FR}_j)}{(1-p^{FR}_j)^2}$ for FR protocol \cite{Najm2017} and $p^{FR}_j$ is given as in \eqref{eq:FR}.

Next steps for the derivation of WI is similar to Appendix~\ref{append_WI}. By the definition of threshold policy, we  solve \eqref{eq:FR_whole} and try to find a $I^{FR}_j(\delta^{rx}_j)\triangleq C$ such that actions of staying idle and starting the transmission to user $j$ are equally desirable, that is, $L_C^{\delta^{rx}_j}$ and $L_C^{\delta^{rx}_j+1}$ give the same result. 
\qed

\subsection{Proof of Theorem~\ref{thm_FIR}}

\label{AppendixJ} 
The proof is similar to that of Theorem~\ref{thm:bound} in Appendix~\ref{append_LB1}. First, constraints of  (i) and (ii) are relaxed and we obtain \eqref{eq:decoupling}. 
$J^*_j$ for FR HARQ protocol can be computed using \eqref{eq:FR_age} and \eqref{eq:FR_transmission}.  Note that both $J_j^{\Gamma_j}$ and $C_j^{\Gamma_j}$ are convex functions of $\Gamma_j$ and  $J_j^{\Gamma_j}$ can be written in terms of $C_j^{\Gamma_j}$
\begin{align}
J_j^{\Gamma_j} = \frac{\Exp{S_j}}{2C_j^{\Gamma_j}} + C_j^{\Gamma_j} \frac{\Var{S_j}}{2\Exp{S_j}}+n_s-\frac{1}{2},\label{lower_boundsingle}\end{align}
where $\Exp{S_j}$ and $\Var{S_j}$ denote the expected value and the variance of $S_j$.  $J_j^{\Gamma_j}$ is a convex increasing function of $C_j^{\Gamma_j}$ and $C_j^* \leq \lambda_j$. Then, by inserting \eqref{lower_boundsingle} into \eqref{eq:decoupling}, similarly to Appendix~\ref{append_WI}, the lower bound for FR HARQ can be computed in closed form as follows:
\begin{subequations}
\begin{align}
J^{*} &\geq \min{\left(\sum_{j=1}^M {w_j\left(\frac{\Exp{S_j}}{2\lambda_j} + \lambda_j \frac{\Var{S_j}}{2\Exp{S_j}}+n_s-\frac{1}{2}\right)}\right)} \label{eq:0FR} \\
&=\frac{1}{2\lambda}{\left(\sum_{j=1}^M {\sqrt{w_j\Exp{S_j}}}\right)}^2+\frac{\lambda}{2}\min_j{ \left(\frac{\Var{S_j}}{\Exp{S_j}}\right)   }+  \sum_{j=1}^M {w_j\left(n_s-\frac{1}{2}\right)},
\end{align}
\end{subequations}
which is equal to the bound in Theorem~\ref{thm_FIR} where $\Exp{S_j}=\frac{n_s}{1-p^{FR}_j}$, and $\Var{S_j}=\Exp{S_j^2}-\Exp{S_j}^2=\frac{n_s^2 p^{FR}_j}{(1-p^{FR}_j)^2}$ for FR protocol.
\qed



\bibliography{thesis1}
\end{document}